\documentclass[11pt]{article}

\usepackage[letterpaper,margin=0.9in]{geometry}
\usepackage{bm, amsmath, amsthm, amssymb, amsfonts}

\usepackage{thmtools,thm-restate}
\usepackage{cite}
\usepackage{framed}
\usepackage[framemethod=tikz]{mdframed}
\usepackage{appendix}
\usepackage{graphicx}
\usepackage{color}
\usepackage{varwidth}
\usepackage{wrapfig}
\usepackage{algorithm}
\usepackage[noend]{algpseudocode}
\usepackage[textsize=tiny]{todonotes}
\usepackage{mathrsfs}  
\usepackage{enumitem}

\usepackage{caption}
\newlength\myindent
\newcommand{\poly}{\operatorname{\text{{\rm poly}}}}

\newcommand{\expect}{\ensuremath{\mathbb{E}}}
\newcommand{\pr}{\ensuremath{\mathbb{P}}}

\renewcommand{\paragraph}[1]{\vspace{0.15cm}\noindent {\bf #1}}

\usepackage{xspace}
\usepackage{xifthen}

\definecolor{darkgreen}{rgb}{0,0.5,0}
\definecolor{darkblue}{rgb}{0,0,0.8}
\usepackage{hyperref}
\hypersetup{
    unicode=false,          % non-Latin characters in Acrobat’s bookmarks
    colorlinks=true,        % false: boxed links; true: colored links
    linkcolor=darkblue,          % color of internal links (change box color with linkbordercolor)
    citecolor=darkgreen,        % color of links to bibliography
    filecolor=magenta,      % color of file links
    urlcolor=brown           % color of external links
}
\RequirePackage[]{silence}
\usepackage{algorithm}
\usepackage[noend]{algpseudocode}

\usepackage[nameinlink,capitalize]{cleveref}

\newtheorem{theorem}{Theorem}[section]

\newtheorem{lemma}[theorem]{Lemma}

\newtheorem{observation}[theorem]{Observation}
\newtheorem*{observation*}{Observation}
\newtheorem*{definition*}{Definition}
\newtheorem{remark}{Remark}

\crefname{theorem}{Theorem}{Theorems}
\Crefname{lemma}{Lemma}{Lemmas}
\Crefname{claim}{Claim}{Claims}
\Crefname{corollary}{Corollary}{Corollaries}
\Crefname{remark}{Remark}{Remarks}
\Crefname{observation}{Observation}{Observations}

\newcommand{\bigO}{O}

\newcommand{\eps}{\varepsilon}

\newcommand{\sh}{\ensuremath{\textrm{SH}_t}}

\newcommand{\orc}{\ensuremath{\mathscr{O}}}

\newcommand{\orci}{\orc_{i}}

\newcommand{\Gtt}[2]{\ensuremath{H_{#1, #2}}}
\newcommand{\Gt}{\ensuremath{\Gtt{t}{t'}}}
\newcommand{\Gi}{\ensuremath{H_{i}}}
\newcommand{\Gs}{\ensuremath{H_{s}}}

\title{Sparsifying Distributed Algorithms \\ with Ramifications in Massively Parallel Computation \\ and Centralized Local Computation}
\author{
 Mohsen Ghaffari\\
  \small ETH Zurich \\
  \small ghaffari@inf.ethz.ch
\and
Jara Uitto \\
\small ETH Zurich \& U. of Freiburg \\
\small jara.uitto@inf.ethz.ch
 }

\date{}

\begin{document}

\maketitle

\begin{abstract}
We introduce a method for ``sparsifying" distributed algorithms and exhibit how it leads to improvements that go past known barriers in two algorithmic settings of large-scale graph processing: Massively Parallel Computation (MPC), and Local Computation Algorithms (LCA).

\medskip
\begin{itemize}
\item\textbf{MPC with Strongly Sublinear Memory:} Recently, there has been growing interest in obtaining MPC algorithms that are faster than their classic $O(\log n)$-round parallel (PRAM) counterparts for problems such as Maximal Independent Set (MIS), Maximal Matching, $2$-Approximation of Minimum Vertex Cover, and $(1+\eps)$-Approximation of Maximum Matching. Currently, all such MPC algorithms require memory of $\tilde{\Omega}(n)$ per machine: Czumaj et al. [STOC'18] were the first to handle $\tilde{\Omega}(n)$ memory, running in $O((\log\log n)^2)$ rounds, who improved on the $n^{1+\Omega(1)}$ memory requirement of the $O(1)$-round algorithm of Lattanzi et al [SPAA'11]. We obtain $\tilde{O}(\sqrt{\log \Delta})$-round MPC algorithms for all these four problems that work even when each machine has strongly sublinear memory, e.g., $n^{\alpha}$ for any constant $\alpha\in (0, 1)$. Here, $\Delta$ denotes the maximum degree. These are the first sublogarithmic-time MPC algorithms for (the general case of) these problems that break the linear memory barrier.

\medskip

\item\textbf{LCAs with Query Complexity Below the Parnas-Ron Paradigm:} Currently, the best known LCA for MIS has query complexity $\Delta^{\bigO(\log \Delta)} \poly(\log n)$, by Ghaffari [SODA'16], which improved over the $\Delta^{\bigO(\log^2 \Delta)} \poly(\log n)$ bound of Levi et al. [Algorithmica'17]. As pointed out by Rubinfeld, obtaining a query complexity of $\poly(\Delta\log n)$ remains a central open question. Ghaffari's bound almost reaches a $\Delta^{\Omega\left(\frac{\log \Delta}{\log\log \Delta}\right)}$ barrier common to all known MIS LCAs, which simulate a distributed algorithm by learning the full local topology, \`{a} la Parnas-Ron [TCS'07]. There is a barrier because the distributed complexity of MIS has a lower bound of $\Omega\left(\frac{\log \Delta}{\log\log \Delta}\right)$, by results of Kuhn, et al. [JACM'16], which means this methodology cannot go below query complexity $\Delta^{\Omega\left(\frac{\log \Delta}{\log\log \Delta}\right)}$. We break this barrier and obtain an LCA for MIS that has a query complexity $\Delta^{\bigO(\log\log \Delta)} \poly(\log n)$.
\end{itemize}
 \end{abstract}

\setcounter{page}{0}
\thispagestyle{empty}
\newpage

\section{Introduction and Related Work}
We introduce a notion of \emph{locality volume} for local distributed algorithms and we show that, by devising local graph algorithms that have a small locality volume (we refer to these as \emph{sparse} algorithms), we can obtain significant improvements in two modern computational settings: Massively Parallel Computation (MPC) and Local Computation Algorithms (LCA). Both of these settings, which are receiving increasingly more attention, are primarily motivated by the need for processing large-scale graphs. We hope that the study of sparse local algorithms and the methodology set forth here may also find applications in a wider range of computational settings, especially for large-scale problems, where ``local approaches" provide a natural algorithmic line of attack.   

\paragraph{The $\mathsf{LOCAL}$ model and the locality radius:} Distributed graph algorithms have been studied extensively since the 1980s. The standard model here is Linial's $\mathsf{LOCAL}$ model\cite{linial1987LOCAL}: the communication network of the distributed system is abstracted as an $n$-node graph $G=(V, E)$, with one processor on each node, which initially knows only its own neighbors. Processors communicate in synchronous message passing rounds where per round each processor can send one message to each of its neighbors. The processors want to solve a graph problem about their network $G$ --- e.g., compute a coloring of it --- and at the end, each processor/node should know its own part of the output, e.g., its color. 

The focus in the study of $\mathsf{LOCAL}$ model has been on characterizing the round complexity of graph problems. This not only captures the time needed by a distributed system to solve the given graph problem, but also characterizes the \emph{locality radius} of the problem, in a mathematical sense: whenever there is an algorithm with round complexity $T$, the output of each node $v$ is a function of the information residing in nodes within distance $T$ of $v$, and particularly the topology induced by the $T$-hop neighborhood of $v$. Thus, in this sense, the problem has locality radius at most $T$.      

\paragraph{The locality volume:} We initiate the study of $\mathsf{LOCAL}$ algorithms that, besides having a small locality radius, also have a small {\emph{locality volume}: in a rough sense\footnote{We note that a precise definition of the locality volume can be somewhat subtle. Instead of providing a cumbersome and detailed mathematical definition, we will explain this notion in the context of a warm up provided in \Cref{sec:warmup}. }, we want that each part of the output should depend on only a few elements of the input, i.e., nodes and edges (and the randomness used to decide about them). In particular, the output of a node $v$ should depend on a small part of the topology within the $T$-hop neighborhood of $v$, instead of all of it. This opens the road for us to devise improved algorithms in MPC and LCA. On a high level, this locality volume will correspond to the memory requirement in the MPC setting (per node) and also to the query complexity in the LCA model. To make this point concrete, we next discuss each of these two settings separately and state our results.

\subsection{Massively Parallel Computation (MPC)}
Massively Parallel Computation (MPC) is a theoretical abstraction which is intended to model recent large-scale parallel processing settings such as MapReduce\cite{dg04}, Hadoop\cite{White:2012}, Spark\cite{ZahariaCFSS10}, and Dryad\cite{Isard:2007}. This model was introduced by Karloff et al.\cite{KarloffSV10} and is receiving increasingly more attention recently\cite{KarloffSV10, goodrich2011sorting, LattanziMSV11, Beame13, Andoni:2014, Beame14, hegeman2015lessons, AhnGuha15,Roughgarden16,Im17,czumaj2017round, assadi2017simple,assadi2017coresets,ghaffari2018improved,harvey2018greedy,brandt2018breaking,assadi2018massively,boroujeni2018approximating,andoni2018parallel}.

\paragraph{The MPC model:} The MPC model consists of a number of machines, each with $S$ bits of memory, who can communicate with each other in synchronous rounds on a complete communication network. Per round, each machine can send $O(S)$ bits to the other machines in total, and it can also perform some local computation, ideally at most $\poly(S)$. For graph problems, the number of machines is assumed to be $\tilde{O}(m/S)$, where $m$ denotes the number of edges in the graph, so that the graph fits the overall memory of the machines. The main objective is to obtain MPC algorithms that have a small \emph{round complexity} as well as a small \emph{memory} per machine.

\paragraph{State of the Art:} In this paper, our focus will be on some fundamental graph problems such as maximal independent set, maximal matching, $(1+\eps)$-approximation of maximum matching, and $2$-approximation of minimum vertex cover. For all of these problems, classic parallel or distributed algorithms imply $O(\log n)$ round MPC algorithms without any serious memory requirement (as long as each node's edges can fit in one machine)\cite{II86, luby1986simple, alon1986fast, lotker2015improved}. Given the power of the MPC model and also the pressing need for fast processing of large-scale graphs, the objective in MPC is to obtain algorithms that are considerably faster than their classic parallel counterparts --- i.e., strongly sublogarithmic time for the above four problems --- using a small memory per machine.   

\paragraph{The Linear Memory Barrier:} The memory requirement for the above four problems has improved over time. Currently, sublogarithmic-time algorithms are known only when the memory $S$ per machine is at least $\tilde{\Theta}(n)$. In fact, this itself became possible only recently, due to a breakthrough of Czumaj et al.\cite{czumaj2017round}: they presented an MPC algorithm with $S=\Theta(n)$ and round complexity $O((\log\log )^2)$ for $(1+\eps)$-approximation of maximum matching. Two independent follow up work provided some improvements: Assadi et al. \cite{assadi2017coresets} obtained an $O(\log\log n)$ round algorithm for $1+\eps$ approximation of maximum matching and $O(1)$-approximation of minimum vertex cover; and Ghaffari et al.\cite{ghaffari2018improved} obtained $O(\log\log n)$ round algorithms for maximal independent set, $(1+\eps)$-approximation of maximum matching, and $(2+\eps)$-approximation of minimum vertex cover. Before this burst of developments for the setting where $S=\tilde{\Theta}(n)$, the best known algorithms were those of Lattanzi et al.\cite{LattanziMSV11} which require memory $S=n^{1+\Omega(1)}$ and have round complexity $O(1)$.

However, all currently known techniques in MPC algorithms for the above four problems lose their efficacy once the memory per machine becomes (strongly) sublinear, e.g., $S=n^{\alpha}$ for a constant $\alpha \in (0, 1)$. In particular, as soon as the memory per machine goes below, say $n^{0.99}$, the best known round complexity for general graphs\footnote{We are aware of one exception for special graphs: For trees, a recent work of Brandt et al.~\cite{brandt2018breaking} obtains an $O((\log\log n)^3)$-round MIS algorithm in the MPC model with memory $n^{\alpha}$ per machine for any constant $\alpha \in (0, 1)$. More recently, they~\cite{brandt2018Arb} generalized this to any graph of arboricity $\poly(\log n)$ and improved the round complexity to $O((\log\log n)^2)$ and this extension also works for maximal matching.} goes back to the $O(\log n)$-round solutions that follow from the classic distributed/parallel algorithms. This is rather unfortunate because this regime of memory--- e.g., $S\leq n^{0.99}$---is especially of interest, as the graph sizes are becoming larger and larger.

\medskip
\paragraph{Our Result in MPC:} By devising sparse $\mathsf{LOCAL}$ algorithms (which have small locality volume), we obtain MPC algorithms that break this barrier. In particular, these algorithm use a strongly sublinear memory per machine and still run considerably faster than $O(\log n)$:

% \ju{dont forget to add references to the claims that prove this theorem}
\begin{theorem} There are MPC algorithms, with memory per machine of $S=n^{\alpha}$ for any constant $\alpha\in (0, 1)$, that, with probability at least $1 - 1/n^{10}$, solve the following four problems in $O(\sqrt{\log \Delta} \cdot \log\log \Delta + \sqrt{\log\log n}) = \tilde{O}(\sqrt{\log \Delta})$ rounds in any $n$-node graph of maximum degree at most $\Delta$: Maximal Independent Set, Maximal Matching, $(1+\eps)$-Approximation of Maximum Matching for any constant $\eps>0$, and $2$-Approximation of Minimum Vertex Cover. 
\label{thm:mainMPC}
\end{theorem} 
We comment that in the case of maximal matching and $2$-approximation of minimum vertex cover, previously there was no known sublogarithmic-time algorithm even for a memory of $S=\tilde{\Theta}(n)$. 

% \mtodo{complexity $\tilde{O}(\frac{\log \Delta}{\sqrt{\log n}})$?}

\paragraph{The Connection to Sparse Local Distributed Algorithms.} We obtain \Cref{thm:mainMPC} by devising $\mathsf{LOCAL}$ algorithms for these problems that have both small locality radius of $O(\log \Delta)$ and also a small locality volume of $\Delta^{O(\sqrt{\log \Delta})}$, in a rough sense. There are also some smaller clean up steps, which we discard from our discussion for now. More concretely, this algorithm will be such that running every span of $R=\Theta(\alpha\sqrt{\log \Delta})$ rounds of it has locality volume at most $\Delta^{\alpha/10}$ for a desirably small constant $\alpha \in (0, 1)$. This $\Delta^{\alpha/10}$ volume fits well within the memory of one machine. In fact, if we uniformly spread the nodes among the machines, each machine has enough memory to store $\Delta^{\alpha/10}$ bits for each of the nodes that it holds (some care is needed when the graph is highly irregular). Using a simple and by now well-known graph exponentiation idea (see, e.g., \cite{Lenzen:2010, Ghaffari2017, brandt2018breaking, andoni2018parallel, assadi2018massively}), we can make each node $v$ learns this $\Delta^{\alpha/10}$ local volume that determines its behavior for the $R$ rounds, within $O(\log R)$ MPC rounds, after which it can locally emulate its behavior for $R$ rounds. Hence, once we have a sparse distributed algorithm where the locality volume fits the memory of a machine, it is easy to compress the number of rounds exponentially. In particular, we can ``compress" each phase of $R=\Theta(\sqrt{\log \Delta})$ rounds of the $\mathsf{LOCAL}$ algorithm into $O(\log R)=O(\log\log \Delta)$ rounds of MPC. Hence, by doing this for different phases, one after the other, we get an $\tilde{O}(\sqrt{\log \Delta})$ round MPC algorithm.

As a side remark, we note that some of the ideas that we use for our sparse local algorithm are similar to those that were used before in \cite{Ghaffari2017}. A particular idea that we borrow from \cite{Ghaffari2017} allows us to effectively \emph{stall} nodes in ``dense" neighborhoods in the MIS algorithm of \cite{Ghaffari2016}, without losing its guarantees. See \Cref{subsec:sparsifying}. The work of \cite{Ghaffari2017} obtains a faster MIS algorithm in $\mathsf{CONGESTED}$-$\mathsf{CLIQUE}$ model of distributed computing, in which the $n$ nodes of the network can communicate with each other in an all-to-all manner, where per round each two nodes can exchange $O(\log n)$ bits. There, the nodes have no memory constraints. Recently, the round complexity of MIS in that model was improved to $O(\log\log n)$ \cite{ghaffari2018improved}, using a very different method.

%\newpage
\subsection{Local Computation Algorithms (LCA)}
\emph{Local Computation Algorithms} (LCA) is a recent theoretical model that was introduced by Rubinfeld et al.\cite{rubinfeld2011fast} and Alon et al.\cite{alon2012LCA}, also motivated by the necessity to process massive graphs. For general introductions, we refer the reader to a comprehensive and recent survery of Levi and Medina\cite{levi2017centralized}. The LCA model is known to be closely related to many other computational models, cf. Levi et al.~\cite[Section 1]{levi2017LCA}, and is stipulated to be useful also in settings such as cloud computing. The high-level goal in this model is to be able to determine each single part of the output in a graph problem in considerably sublinear time, by reading only a few places in the graph. 

Concretely, an LCA has query access to a graph $G=(V, E)$ where each query can ask either for the degree of a node $v\in V$ or for the identifier of the $i$-th neighbor of a node $v$. 
In this work, we assume that a query to a node returns the identifiers of all its neighbors and point out that this only adds a $\Delta$ factor to the query complexity.
It also has access to a string of $n \poly \log n$ random bits\footnote{The number of bits can reduced to $\poly \log n$ using techniques from~\cite{alon2012LCA, levi2017LCA}; we defer the details to the full version of this paper.}. An LCA should be able to determine each single part of the output. For instance, in the Maximal Independent Set (MIS) problem, when asked about a node $v\in V$, the LCA should determine whether $v$ is in the MIS or not, using a small number of queries. All the answers of the algorithm for different vertices $v$ should be consistent with one MIS. 

A central problem in the study of LCAs is that of computing an MIS. This centrality is in part due to fact that many other local problems can be solved using MIS algorithms. This includes maximal matching, $2$-approximation of minimum vertex cover, $(\Delta+1)$-vertex-coloring of graphs of max degree at most $\Delta$, $(2\Delta-1)$-edge coloring, and $(1+\eps)$-approximation of maximum matching \cite{luby1986simple, lotker2015improved, even2014deterministic}.

\paragraph{State of the Art on LCAs for MIS:} Much of the known MIS LCAs are efficient only for graphs of small degrees. In general, the query complexity of known algorithms is a function of two parameters, the maximum degree $\Delta$ and the number of nodes in the graph $n$. Rubinfeld et al.\cite{rubinfeld2011fast} and Alon et al.\cite{alon2012LCA} presented algorithms with query complexity $2^{\bigO(\Delta \log^2 \Delta)}  \log n$. Reingold and Vardi~\cite{reingold2016new} gave an algorithm with query complexity $2^{\bigO(\Delta)} \log n \cdot \log\log n$. Even et al.\cite{even2014deterministic} significantly improved the dependency on $n$ at the cost of increasing the $\Delta$-dependency; concretely they provide  a deterministic LCA with query complexity $2^{\bigO(\Delta^2 \log^2 \Delta)} \log^* n$. 

All the above algorithms have an exponential (or higher) dependency on $\Delta$. Thus, these algorithms lose efficacy in graphs with moderately super-constant degrees, e.g., even for $\Delta=\Omega(\log n)$. There are two known LCAs whose complexity has a better dependency on $\Delta$. Levi et al.\cite{levi2017LCA} gave the first such algorithm with query complexity $2^{\bigO(\log^{3} \Delta)} \log^3 n$ and consequently, Ghaffari\cite{Ghaffari2016} gave an algorithm with query complexity $2^{\bigO(\log^{2} \Delta)} \log^3 n$. A natural question which remains open\footnote{This question was alluded to by Rubinfeld in a TCS+ talk, which can be found here: ~\textcolor{blue}{\url{https://www.youtube.com/watch?v=R8J61RYaaDw}}.} is this:
\begin{center}
``\emph{Is there an MIS LCA with query complexity $poly(\Delta\log n)$?}" 
\end{center}
 
\paragraph{A Natural Query-Complexity Barrier?} The $\Delta^{\bigO(\log \Delta)} \poly(\log n)$ complexity of Ghaffari's algorithm\cite{Ghaffari2016} comes close to a natural barrier for known techniques. The current MIS LCAs, including those of \cite{rubinfeld2011fast, alon2012LCA, levi2017LCA, Ghaffari2016}, are all implicitly or explicitly based on transforming $\mathsf{LOCAL}$ distributed algorithms to LCAs. This is a connection that was first observed by Parnas and Ron\cite{parnas2007approximating}. In particular, given a $T$-round $\mathsf{LOCAL}$ algorithm, we can emulate it in the LCA model with query complexity $\Delta^{\bigO(T)}$: upon being queried on a node $v$, we read the whole subgraph within $T$-hops of $v$, which has $\Delta^{\bigO(T)}$ vertices, and then compute the output $v$ by emulating the $\mathsf{LOCAL}$ algorithm. Now, it is known by a lower bound of Kuhn et al.\cite{KuhnMW16} that any $\mathsf{LOCAL}$ algorithm for MIS needs round complexity at least $\Omega\left( \min\left\{\frac{\log \Delta}{\log\log \Delta}, \sqrt{\frac{\log n}{\log\log n}} \right\} \right)$. Hence, unless we go away from the Parnas-Ron methodology, we cannot go below query complexity $\Delta^{\Omega(\log \Delta/\log\log \Delta)}$.

\medskip
\paragraph{Our Result in LCA:} By devising a sparse $\mathsf{LOCAL}$ algorithms (one that has a small locality volume), we obtain an LCA that goes significantly below the aforementioned barrier. Concretely, we show that:

\begin{restatable}{theorem}{LCAmain}
% \begin{theorem} 
There is an LCA that, with probability $1 - 1/n^{10}$, computes an MIS with query complexity $\Delta^{\bigO(\log\log \Delta)} \poly(\log n)$.
	\label{thm:mainLCA}
% \end{theorem} 
\end{restatable}

While this still does not reach the milestone of $\poly(\Delta\log n)$ query complexity, it makes a significant step in that direction. In particular, in terms of dependency on $\Delta$, it exhibits an exponential improvement in the exponent, compared to the $\Delta^{\bigO(\log \Delta)} \poly(\log n)$-query LCA of Ghaffari\cite{Ghaffari2016}. 

\paragraph{The Connection to Sparse Local Distributed Algorithms.} We obtain \Cref{thm:mainLCA} by devising a $\mathsf{LOCAL}$ MIS algorithm that has a small locality volume of $\Delta^{O(\log \log \Delta)}\poly(\log n)$, as well as a locality radius $O(\log \Delta)$. While the lower bound of Kuhn et al.\cite{KuhnMW16} shows that the locality radius should be at least $\bigO\left(\frac{\log \Delta}{\log\log \Delta}\right)$, our algorithm shows that we do not need to depend on all of the information within this radius, and a much smaller volume suffices. When the LCA is asked whether a given node $v$ is in the MIS or not, it carefully finds its way through this maze of the $\bigO(\log \Delta)$-neighborhood and gathers the relevant $\Delta^{O(\log \log \Delta)}\poly(\log n)$ local volume, using a proportional number of queries. Then, it can emulate the $\mathsf{LOCAL}$ process and determine the output of $v$.

%\mtodo{Perhaps say a few sentences about why this smaller volume cannot be used in the MPC side}
 
%\paragraph{Roadmap.} In Section~\ref{sec:BeepingMIS}, we explain an algorithm in the LOCAL model of distributed computing, which computes a near-maximal independent set. Indeed, the algorithm will be in a simpler and more rudimentary model known as the beeping model. This distributed algorithm will be the base of our LCA MIS in the sense that we will provide a query-efficient simulation for this distributed algorithm in the LCA model, as described in Section~\ref{sec:simu}. We call this part of the algorithm that finds a near-maximal set the pre-shattering phase, while the rest is called post-shattering. Here, the shattering phenomenon refers to the following effect: The pre-shattering algorithm computes a near-maximal independent set $IS$ in the sense that there might be some nodes that are neither in $IS$ nor adjacent to $IS$. However, the remainder graph induced by these nodes consists of small components of size $O(\Delta^4 \log n)$. Once the graph has \emph{shattered} into these small components, it is easy to augment $IS$ into a Maximal Independent Set by brute-forcing through each small component.

\subsection*{Roadmap} In \Cref{sec:warmup}, as a warm up, we present a sparse distributed algorithm for constant approximation of maximum matching, and we explain how, thanks to its small locality volume, it leads to improvements in MPC and LCA settings. In \Cref{sec:MIS}, we present our main sparse MIS algorithm and discuss its implications in the MPC and LCA settings. In particular, this section provides the proof of \Cref{thm:mainMPC}. In \Cref{sec:ImprovedLCA}, we explain how we improve the query complexity of the LCA presented in \Cref{sec:MIS} further, to prove \Cref{thm:mainLCA}. 

%\newpage
\section{Warm Up: Matching Approximation}
\label{sec:warmup}
In this section, we recall a basic distributed algorithm for constant-approximation of maximum matching, and we explain how, by sparsifying it, we can obtain improvements in Massively Parallel Computation (MPC) and centralized Local Computation Algorithms (LCA). 

Concretely, the basic distributed algorithm has a (near-optimal) round complexity of $O(\log \Delta)$. By sparsifying it, we obtain an algorithm with locality radius $O(\log \Delta)$ and locality volume $\Delta^{O(\sqrt{\log \Delta})}$, which then leads to the following results: (I) an $\tilde{O}(\sqrt{\log \Delta})$-round MPC algorithm with strongly sublinear memory per-machine, i.e., $n^\alpha$ bits for any arbitrary constant $\alpha\in (0, 1)$, and (II) an LCA algorithm with query complexity $\Delta^{O(\sqrt{\log \Delta})}$. While being warm ups, these are already considerable improvements over the state of the art: The former is the first sublogarithmic-time MPC algorithm that can handle sublinear memory. The latter LCA has a query-complexity that goes below that of the Parnas-Ron paradigm~\cite{parnas2007approximating}, i.e., collecting the full topology in the $T$-hop neighborhood where $T$ is the distributed complexity of the problem. 
%We will later present an algorithm with improved query complexity to $\Delta^{O(\log\log \Delta)}$ in \Cref{sec:ImprovedLCA}.\mtodo{Add a forward reference} 

Next, we start with explaining the basic distributed algorithm for approximating maximum matching and then we present the sparsified version of it. Then, we discuss how we simulate this sparsified version in the MPC and LCA settings.

\paragraph{Basic Distributed Algorithm:} The algorithm has $\log \Delta$ iterations indexed by $i\in \{0, 1, 2, \dots, \log \Delta - 1 \}$. We maintain the invariant that in iteration $i$, the maximum degree is at most $d_{i-1} = \Delta/2^{i}$. The $i^{th}$ iteration works as follows: we \emph{mark} each edge incident on any node of degree at least $d_i=\Delta/2^{i+1}$ with probability $p_i=2^{i}/(4\Delta)$. We note that this probability is set such that each node of degree at least $d_i=\Delta/2^{i+1}$ has at least a constant probability of having an \emph{isolated marked} edge --- i.e., a marked edge that has no other marked edge adjacent to it. Then, we put all isolated marked edges into the matching and we remove their endpoints from the graph. We also remove from the graph all vertices of degree at least $d_i=\Delta/2^{i+1}$. 

On an intuitive level, in iteration $i$, we remove a number of vertices linear in the number of vertices of degree at least $d_{i-1}/2$ and expect a constant fraction of these to be matched. Hence, overall, the algorithm gives a constant approximation of maximum matching. The formal analysis will be provided for our sparsified variant, which we discuss below. We also comment that this algorithm can be viewed as a simple variant of the algorithm used by Parnas and Ron paradigm~\cite{parnas2007approximating} and it is also close to some algorithms in \cite{onak2010maintaining, czumaj2017round, ghaffari2018improved}.

\paragraph{Sparse Distributed Algorithm:} We now explain how to sparsify this basic algorithm and make sure that it has a small locality volume. We break the algorithm into $2\sqrt{\log \Delta}$ phases; each with $R=\sqrt{\log \Delta}/2$ iterations. We simulate each phase, by running a different $O(R)$-round distributed algorithm on a sparsified graph $H \subseteq G$. The graph $H$ has maximum degree $\tilde{O}\left(2^{\sqrt{\log \Delta}/2}\right)$ and moreover, it can be identified at the beginning of the phase in one round. 

Let us focus on the first phase; the other phases are similar. For this phase, we generate $H$ randomly, as follows: For each iteration $i\in \{1, 2, \dots, R\}$ to be simulated, generate a randomly sampled subgraph $H_i$ by including each edge of the original graph $G$ with probability $p'_i= \min\{K p_i, 1\}= \min\{K\cdot 2^{i} /(4\Delta), 1\}$ for some $K=\Theta(\log \Delta)$. The samplings for different iterations $i$ are independent, and are all generated at the same time. The sparsified subgraph $H$ is the union of all of these subgraphs, i.e., $H= \cup_{i\in \{1, 2, \dots, R\}} H_i$. 

To simulate iteration $i$ of the basic algorithm by running another algorithm on the sparsified graph $H$, for each iteration $i\in \{1, 2, \dots, R\}$, we do two things: (1) To mark edges of iteration $i$, we subsample each sampled edge of $H_i$ (those whose both endpoints are still present) with probability $p_i/p'_i$. Then, as before, isolated marked edges are added to the output matching, and their vertices are removed from $H_i$. (2) Instead of removing vertices of high degree in the original graph $G$ (which we cannot identify exactly as we do not want to communicate in $G$), we remove all vertices whose remaining degree in $H_i$ exceeds $d_i \cdot p'_i=K/8$. This completes the description for one phase. After a phase on the sparsified graph, we use one round of communication on $G$ to remove all vertices whose degree in the remaining graph exceeds $d_R=\Delta/2^{R+1}$. Then, we proceed to the next phase. Other phases work similarly, essentially as if the maximum degree has decreased by a factor of $2^{R} = 2^{\sqrt{\log \Delta}/2}$. 

\begin{lemma}
(A) For each node $v$, the degree of $v$ in $H$ is at most $2^{\sqrt{\log \Delta}/2} \cdot O(\log \Delta)$, with probability at least $1-1/\Delta^{10}$. (B) After iteration $i$ of the simulation, the remaining degree of each node in graph $G$ is at most $2d_i=\Delta/2^{i}$, with probability at least $1-1/\Delta^{10}$. (C) In iteration $i$ of the simulation, if we remove a set $S$ of vertices (for having a high-degree in $H_i$ or becoming matched), then we have $\Theta(|S|)$ matched edges in this iteration, with probability at least $1-exp(-\Theta(|S|))$. Hence, the algorithm computes a constant approximation of maximum matching, with high probability (in the matching size).
\end{lemma}

\begin{proof} (A) Since $p'_i=K \cdot 2^{i}/(4\Delta)$, the expected degree of $v$ in $H_i$ is at most $2^{i} \cdot O(\log \Delta)$. Thus the expected degree in $H$ is at most $2^{\sqrt{\log \Delta}/2} \cdot O(\log \Delta)$. By a Chernoff bound, the probability that $v$ has more than $2^{\sqrt{\log \Delta}/2} \cdot O(\log \Delta)$ edges in $H$ is no more than $1/\Delta^{10}$.

(B) For any node whose degree is at least $\Delta/2^i$ in graph $G$ in the end of simulation of iteration $r=i-1$, we expect to have at least $\Delta/2^i \cdot 2^{i}/(4\Delta) \cdot K =K/4$ sampled edges in $H_i$ (to vertices that are not removed after simulating iterations $1$ to $i-1$). Hence, with probability at least $1-1/\Delta^{10}$, any such vertex has at least $K/8$ sampled edges in $H_i$ and thus gets removed in iteration $i$.

(C) 
The set $S$ of vertices that get removed in iteration $i$ is composed of two parts: the set of vertices that are incident on isolated marked edges, which are matched, and the set $S'$ of vertices that have a degree of at least $K/4$ in $H_i$.  To show that the matching size is $\Theta(|S|)$, it suffices to show that the matching size is $\Omega(|S'|)$. Now, each edge of $H_i$ gets marked with probability $p_i/p'_i$. Therefore, the probability for each vertex $v\in S'$ (which has at least $K/4$ edges in $H_i$) to be incident on a marked edge is at least a constant. Now, by property (B), at the beginning of iteration $i$, each node $u$ has degree at most $2d_{i}$ in graph $G$, with probability at least $1-1/\Delta^{10}$, which means that in the sampled graph $H_i$, node $u$ has degree at most $4d_i p'_i=K/2$, with probability at least $1-2/\Delta^{10}$. 
Therefore, with probability at least a constant, the marked edge incident on $v$ is isolated and has no adjacent marked edge. That is, with constant probability, vertex $v$ is matched. This implies that we expect to have at least $\Theta(|S'|)$ matched edges in this iteration. By an application of McDiarmid's inequality (similar to \cite[Lemma 4.1]{ghaffari2018improved}), we can also prove some concentration around this expectation, and show that, with probability at least $1-exp(-\Theta(|S'|))$, the matching size is at least $\Theta(|S'|)$.
\end{proof}
\paragraph{The Locality Volume of the Sparsified Graph.} To simulate one phase, we need to identify the related sparsified graph, which can be done in one round and has degree $2^{O(\sqrt{\log \Delta})}$. Then, we run a $\Theta(\sqrt{\log \Delta})$ round local process on this graph. Hence, each node's behavior in one phase depends on at most $\Delta \cdot \left(2^{O(\sqrt{\log \Delta})}\right)^{\Theta(\sqrt{\log \Delta})} = \Delta^{O(1)}$ nodes/edges. This is the locality volume for one phase. Since we have $\Theta(\sqrt{\log \Delta})$ phases, the overall behavior of each node during this algorithm depends on at most $\Delta^{\Theta(\sqrt{\log \Delta})}$ other nodes/edges. Hence, although we have a process with locality radius $T=\Theta(\log \Delta)$, the locality volume is much smaller than $\Delta^{T}$ and is just $\Delta^{\Theta(\sqrt{\log \Delta})}$.

%\mtodo{Say something about the locality volume}

\paragraph{Simulation in the MPC Model:}
%Before describing the simulation, we comment that for the purpose of this warm-up, we assume that $\Delta\leq n^{\alpha}$. This assumption is only for convenience and can be removed by shortening the phase length slightly, as is done in our main result in \Cref{sec:MIS}. 
%
Before describing the simulation, we comment that for this warm-up, and to avoid technicalities, we make two simplifying assumption: (1) We assume that the edges of each node fit within one machine (thus, $\Delta \leq n^{\alpha}$). This assumption can be avoided, basically with some change of parameters, as is done for our main algorithms presented in the next section. (2) We assume that we have room for at least $\tilde{O}(\sqrt{\Delta})$ bits per node, e.g., by working under the assumption that all vertices have degrees between $[\Delta^{1/2}, \Delta]$, which means that the number of machines $M= \tilde{\Omega}\left( \frac{n\sqrt{\Delta}}{n^{\alpha}} \right)$ and thus, $n/M \geq \tilde{O}\left( \frac{n^{\alpha}}{\sqrt{\Delta}} \right)$. This latter assumption can be removed by working through $\log\log \Delta$ successive iterations of polynomially decreasing degree classes $[\Delta^{1/2}, \Delta]$, $[\Delta^{1/4}, \Delta^{1/2}]$, $[\Delta^{1/8}, \Delta^{1/4}]$, etc.
 
We simulate $R=\sqrt{\log \Delta}/2$ iterations of one phase in the sparsified algorithm in $O(\log \log \Delta)$ rounds of the MPC model. For that, we make each node $v$ learn its $R$-hop neighborhood in $H$, as follows: we have $\log R$ MPC rounds, where at the end of round $i$, each node should know its neighborhood in $H^{2^i}$. In round $i+1$, node $v$ sends the names of all its neighbors in $H^{2^i}$ to all of these neighbors. Given the degree of $H$, this is at most $(2^{\sqrt{\log \Delta}/2})^{2^i} \leq (2^{\sqrt{\log \Delta}/2})^{R} = \Delta^{1/4}$ neighbors. Hence, at the end of round $i+1$, each node knows its neighbors in $H^{2^{i+1}}$. After $\lceil\log R\rceil$ rounds, each node knows its $R$-neighborhood in $H$. Notice that each machine needs to gather at most $\tilde{O}(\Delta^{1/2})$ bits for each of the $n/M= \tilde{O}(\frac{n^{\alpha}}{\sqrt{\Delta}})$ nodes that it wants to simulate. Hence, each machine can gather this information for all of its nodes and that would fit within its memory. At this point, the machine can locally simulate the behavior of each of its nodes $v$ in $R$ rounds of the algorithm and learn whether $v$ is matched or not and whether it is removed or not. We can then use one round of the MPC model to remove all vertices whose degree has not dropped below $\Delta/2^{R}$, at which point we can proceed to the next phase.

We should remark about one small subtlety in this simulation: We want that the collected neighborhood includes the related random values, so that the simulation (and particularly subsampling) performed after collecting the local topology is consistent in various vertices that simulate the algorithm. For that, we do as follows: for each edge in $G$, when sampling it for inclusion in $H_i$, we draw a uniformly random number in $[0, 1]$. If this random number exceeds $p'_i$, the edge is included in $H_i$. Then, we also include this random number in the information of that edge. When simulating iteration $i$, where we want to subsample and mark edges of $H_i$ with probaility $p_{i}/p'_{i}$, we call each edge of $H_i$ marked if its random number exceeds the threshold of $p_i$.

\paragraph{Simulation in the LCA Model:} 
We start with discussing the simulation of the first phase. We can create an oracle that simulates the $R=\sqrt{\log \Delta}/2$ iterations of this phase for one node $v$ in the LCA model, as follows: We will basically gather $R$-hop topology of $v$ in the sampled graph $H$. This is a topology of size at most $\Delta^{1/2}$ as argued above. For that, we need to build $H$, which we will do iteratively: We first determine all edges of $H$ that are incident on $v$. That will take $\Delta$ queries, to read all neighboring edges, and to sample them according to the probabilistic construction of $H$. Then, we recurse among the at most $O(\log^2 \Delta) \cdot 2^{\sqrt{\log \Delta}/2}$ neighbors of $v$ in $H$, and build their neighborhoods. We then continue on their neighbors, and so on, up to distance $R$. Building edges of each node takes $\Delta$ queries, to determine its edges (and sample the respective random variables), and we then continue on at most $O(\log^2 \Delta) \cdot 2^{\sqrt{\log \Delta}/2}$ neighbors. Since we do this for $R$-hop neighborhood, we need at most $\Delta \cdot (O(\log^2 \Delta) \cdot 2^{\sqrt{\log \Delta}/2})^R \leq \Delta^{3/2}$ queries. Hence, this oracle can simulate one node's behavior in one phase of the sparsified distributed algorithm using $\Delta^{3/2}$ queries. Then, the process for simulating the second phase is similar, except that to simulate each node's behavior in the second phase, we first need to call the oracle of the first phase on this node and its neighbors to know this node's status at the end of that period (the neighbors are needed so that we can remove the node if its degree did not drop below $\Delta/2^{R}$). Hence, the oracle of the second phase works in $\Delta^{5/2}$ queries to the first oracle, which is at most $\Delta^5$ queries to the base graph. Similarly, we can simulate all the $\log \Delta$ iterations in $2\sqrt{\log \Delta}$ phases, where the oracle of each phase makes $\Delta^{5/2}$ calls to the oracle of the previous phase, and at the very base, the calls are to the original graph. Since we have $2\sqrt{\log \Delta}$ phases, the overall query complexity for simulating each node's behavior in the full run of the algorithm is $\Delta^{O(\sqrt{\log \Delta})}$.  

%\newpage
%\input{beepingMIS}

%\newpage

\section{Maximal Independent Set (MIS) and Implications}
\label{sec:MIS}
Here, we first review a distributed MIS algorithm of Ghaffari~\cite{Ghaffari2016} and then present a sparsification for it. We then explain how this sparsification leads to improved MPC and LCA algorithms. 
%\mtodo{Write about the implications on other problems}

\subsection{Reviewing Ghaffari's MIS Algorithm} The MIS algorithm of~\cite{Ghaffari2016} is basically repeating a simple $O(1)$-round probabilistic dynamic, as presented below. Running this dynamic for $O(\log n)$ iterations computes a Maximal Independent Set of the graph, with probability at least $1-1/\poly(n)$. If we run the dynamic for just $O(\log \Delta)$ iterations instead, with probability at least $1-1/\poly(n)$, we obtain a Nearly-Maximal Independent Set, in the following sense: with high probability, the number of remaining nodes is at most $n/\poly(\Delta)$, and each remaining component has size $O(\Delta^4 \log n)$. These two properties allow us to complete the computation of MIS, simply by computing an MIS among remaining nodes, much easier.   

\begin{algorithm}
	\caption{Ghaffari's Local MIS algorithm for node $v$:}
	\label{alg:ghaffari}
	\begin{algorithmic}
		\State Set $p_0(v) = 1/2$. 
		\For{iteration $t = 1, 2, \ldots$ until node $v$ is removed} 
				\State {\bf Round 1:} Set
					\[
						p_{t}(v) = 	\begin{cases}
											p_{t - 1}(v)/2, \qquad \hfill \; \; \; \; \; \;  \; \; \; \; \; \; \; \;\text{if } d_{t - 1}(v) = \sum_{u \in N(v)} p_{t-1}(u) \geq 2, \\
											\min\{ 2p_{t - 1}(v), 1/2 \}, \hfill \text{otherwise.}
									\end{cases}
					\]
				\State {\bf Round 2:} Node $v$ marks itself w.p. $p_{t}(v)$.
				\State If $v$ is the only marked node in its neighborhood $N(v)$, then $v$ joins the MIS. 
				\State If $v$ joined the MIS, $v$ is removed from the graph along with its neighbors.
		\EndFor
	\end{algorithmic}
\end{algorithm} 

\medskip
\paragraph{Intuitive Discussion of How This Algorithm Works:} Informally, the dynamic adjustments in the probabilities $p_{t}(v)$ aim to create a negative-feedback loop so that we achieve the following property: for each node $v$, there are many iterations $t$ in which either (I) $p_t(v)=\Omega(1)$ and $d_t(v)=O(1)$, or (II) $d_t(v) =\Omega(1)$ and a constant fraction of it is contributed by neighbors $w$ for which $d_t(w)=O(1)$. These are good iterations because it is easy to see that in any such iteration, node $v$ gets removed with at least a constant probability. Ghaffari's analysis~\cite{Ghaffari2016} shows that if we run for $\Omega(\log \Delta)$ iterations, each node $v$ spends a constant fraction of the time in such good iterations (with a deterministic guarantee). Hence, if we run for $O(\log n)$ rounds, with high probability, we have computed an MIS. Running for $O(\log \Delta)$ rounds leaves each node with probability at most $1/\poly(\Delta)$ and this can be seen to imply that we have computed a Nearly-Maximal Independent Set~\cite{Ghaffari2016}, in the sense explained above. After that, it is easier to add some more vertices to this set and ensure that we have an MIS. 

\subsection{Sparsifying Ghaffari's MIS Algorithm}
\label{subsec:sparsifying}
\paragraph{Intuitive Discussions About Sparsification}\footnote{We note that the discussions here are quite informal and imprecise. We still provide this intuitive explanation with the hope that it delivers the main idea behind our approach, and why we do certain potentially strange-looking things. }: We are mainly interested in running $O(\log \Delta)$ rounds of the above algorithm; after that we can complete the computation from a near-maximal IS to a maximal IS easier. We will break the algorithm into phases and perform a sparsification for each phase separately. For one phase, which has $R$ rounds, we would like to devise a much sparser graph $H$ such that by running a distributed algorithm on $H$ for $\Theta(R)$ rounds, we can simulate $R$ iterations of Ghaffari's algorithm on the base graph $G$. In our case, we will be able to do this for $R = O(\sqrt{\log \Delta})$. Thus, each $O(\sqrt{\log \Delta})$ iterations can be performed on a much sparser graph and we just need to ``stitch together" $O(\sqrt{\log \Delta})$ of these, by communications in the base graph. We next discuss the challenges in sparsifying one phase and our ideas for going around these challenges. 

We discuss how we deal with sparsification for the first round of iterations, i.e., the round of updating probabilities $p_{t}(v)$ based on the neighbors. We use a similar idea for the sparsification needed for the second rounds of the iterations, where we perform a marking to determine the vertices that are added to MIS.
 
Let us first examine just one round of the dynamic. One obstacle is the dynamic update of the probabilities $p_{t}(v)$, which depend on all the neighbors. That is, $p_{t}(v)$ is updated based on the summation of the probabilities $p_{t}(u)$ of all neighbors $u \in N(v)$. It seems like even if we ignore just one or a few of the neighbors, and we do not include them in $H$, then the update of the probability might be incorrect, especially if those ignored neighbors $u$ have a large value $p_{u}(t)$. However, all that we need to do is to test whether $d_{t - 1}(v) = \sum_{u \in N(v)} p_{t-1}(u) \geq 2$ or not. Thus, a natural idea for sparsification is to use random sampling, while neighbors of larger $p_{t}(u)$ have more importance. In particular, if we sample each node $u$ with probability $p_{t}(u)$ and compare the number of sampled vertices with $2$, we have a constant-probability random tester for checking the condition $d_{t - 1}(v) = \sum_{u \in N(v)} p_{t-1}(u) \geq 2$, up to a small constant factor. That is, if we are above the threshold by a constant factor, the test detects that we are above the threshold with at least a positive constant probability, and if we are below the threshold by a constant factor, the test detects that we are below the threshold with at least a positive constant probability. We can run this random tester several times, all in parallel, to amplify the success probability. For each node $v$, the sampled set of neighbors would have size at most $\poly(\log \Delta) \cdot d_{t-1}(v)$, with probability $1-1/\poly(\Delta)$, thus opening the way for the creation of the sparser graph $H$ mentioned above, especially for nodes $v$ whose $d_{t-1}(v)$ is small. 

The above does not seem so useful on its own, because we still have to receive from each neighbor $u$ whether it is sampled or not, and that depends on $p_{t}(u)$ in that iteration which is not known in advance. But, thanks to the fact that the changes in $p_{t}(u)$ are somewhat smooth, we can go much further than one round, as we informally sketch next. Suppose that at time $t$, which is the beginning of a phase, we want to build a sparse graph $H$ that includes any neighbor that may be sampled and thus might impact the estimation of $d_{t'}(u)$ in any iteration $t'\in [t, t+R]$. For each round $t'\in [t, t+R]$, if we include each node with probability $2^{R} p_{t}(v)$, the included set would be an oversampling of the set that will be sampled at iteration $t'\in [t, t+R]$, i.e., it will include the latter. This is because $p_{t'}(v) \leq 2^R p_{t}(v)$. The fact that at time $t$ we can predict a small superset of all vertices that will be sampled in iterations $[t, t']$ allows us to build a graph $H$ where each node $v$ has at most $\tilde{O}(2^{R}d_{t}(v))$ neighbors, and suffices for simulating the next $R$ rounds. We soon discuss how to deal with vertices for which $d_{t}(v)$ is large.

The above randomly sampled graph $H$ is good for vertices $v$ such that $d_{t}(v)$ is small, e.g., $2^{\tilde{O}(R)}$. But for vertices that have a larger $d_{t}(v)$, this graph would include many neighbors, which is not desired. Fortunately, for any such vertex $v$ for which $d_{t}(v) \geq 2^{3R}$, we have a different nice property, which helps us predict their behavior for the next $R$ rounds. More correctly, this property enables us to safely gamble on a prediction of their behavior. 

Let us explain that: Under normal circumstances where for each neighbor $p_{t}(u)$ decreases by a $2$ factor per round, during the next $R$ round, $d_{t}(v)$ would decrease by at most a $2^R$ factor. Hence, if we start with $d_{t}(v) \geq 2^{3R}$, throughout all iterations $t'\in [t, t+R]$ in the phase, $d_{t'}(v)$ is quite large. In such cases, it is clear that $v$ should keep reducing its $p_{t}(v)$ and also that any time that it marks itself, it gets blocked by a marked neighbor, with a significant probability. Hence, in such an situation, the behavior of $v$ is predictable for the next $R$ rounds. Of course, it is possible that many of the neighbors of $v$ drop out during the next $R$ rounds and because of that we suddenly have $d_{t'}(v)\leq 2^R$. Fortunately, this is enough progress in the negative-feedback dynamic around $v$, which allows us to modify the analysis of \cite{Ghaffari2016} and show the following property. The algorithm works even with the following update: if at the beginning of the phase we have $d_{t}(v)\geq 2^{3R}$, for all rounds of this phase, we can update $p_{t'+1}(v)=p_{t'}(v)/2$ without checking $d_{t'}(v)$ (in a very predictable manner). In this case, we say that we are \emph{stalling} node $v$. In a sense, this postpones the attempts of $v$ to join MIS for the next $R$ rounds. On an intuitive level, this is fine because sudden drops that $d_{t}(v)\geq 2^{3R}$ and $d_{t'}(v) \leq 2$ for some $t' \in [t, t+R]$ cannot happen too frequently. We note that an idea similar to this was used before in \cite{Ghaffari2017} to obtain an algorithm for MIS algorithm in the $\mathsf{CONGESTED}$-$\mathsf{CLIQUE}$ model of distributed computing.

Finally, we note that in the above, we discussed our idea for randomly testing whether $d_{t'}(v)\geq 2$ or not, via randomly sampling vertices. Essentially the same idea can be used to create a superset of marked nodes, such that it has only a few nodes around each node $v$ whose $d_{t}(v)$ is small and it is guaranteed to include all neighbors of $v$ that are marked in round $t'$. 

%The difference to the algorithm by Ghaffari lies in the stalling process.
%The intuition is that if $d_t(v) \geq 2^{2x} \gg 2^x \cdot \log \Delta$, it is more likely that node $v$ will not join the MIS in the next $x$ rounds because there will be at least one marked neighbor in every iteration. Therefore, node $v$ might as well give up and not try to mark itself. The crucial observation is that stalling makes the behavior of $v$ predictable for the duration of the stalling and therefore, we can simulate a node that is stalling without any communication. 

\paragraph{Sparsified Variant of Ghaffari's Local MIS Algorithm.} The precise algorithm can be found in Algorithm~\ref{alg:localmis}.  
Let us summarize the changes to Ghaffari's algorithm: As mentioned above, we break the algorithm into phases, each made of $R=\alpha\sqrt{\log{\Delta}}/10$ iterations, and we do the sparsification mentioned for each iteration. Recall that $\alpha\in (0, 1)$ is the constant so that each machine has memory at least $n^{\alpha}$. At the beginning of the phase, we decide whether to stall each node $v$ or not, based on the value of $d_{t}(v)$ at that point. Furthermore, instead of updating $p_{t}(v)$ by reading the summation of all neighbors, we update it based on an estimation that derived from $O(\log \Delta)$ parallel repetitions of sampling each neighbor $u$ with probability $p_{t}(u)$. 
%
%The precise algorithm that incorporates these changes can be found in Algorithm~\ref{alg:localmis}.

\begin{algorithm}
	\caption{Local MIS algorithm for node $v$:}
	\label{alg:localmis}
	\begin{algorithmic}
		\State Set $p_0 = 1/2$.
		\For{phase $s = 0, 1, \ldots$ until node $v$ is removed}			
% 			\State If $\hat{d}_{R}(v) \geq 2^{\sqrt{\log \Delta} / 5}$, then \emph{stall} for this phase, i.e., the next $\sqrt{\log \Delta} / 10$ iterations. \vspace{+5pt}
			\For{iteration $i = 1, 2, \ldots, \alpha \cdot \sqrt{\log \Delta} / 10$ of phase $s$} 		\Comment{$n^{\alpha}$ is the memory per machine}			
				\State Let $t = s \cdot \alpha \cdot \sqrt{\log \Delta} / 10 + i$ and let $k = 12 \cdot C \log \Delta$. \Comment{$C$ is some large constant}
				\State Perform $k$ repetitions of sampling, where in each repetition $v$ is sampled w.p. $p_{t - 1}(v)$.
				\State Let $b(v)$ be the binary vector of length $k$, where $b^{\,j}$ is its $j$-th element.
				\State Set $b^{\,j}$ equal to $1$ iff $v$ is sampled in repetition $j$.
				\State Let $\hat{N}(v) \subseteq N(v)$ be the set of neighbors sampled at least once. 
				\State For $j = 1, \ldots, k$, set
					\[
						\hat{d}^{\,j}(v) = \sum_{u \in \hat{N}(v)} b^{\,j}(v) \ .
					\] 
				\State Set estimate $\hat{d}_{t - 1}(v)$ as the median of $\{\hat{d}^{1}, \hat{d}^{2}, \ldots, \hat{d}^{k}\}$.
					
				\State If $i = 1$ and $\hat{d}_{t - 1}(v) \geq 2^{\alpha \cdot \sqrt{\log \Delta} / 5}$, then stall for this phase.
				\State {\bf Round 1:} Set
					\[
						p_{t}(v) = 	\begin{cases}
											p_{t - 1}(v)/2, \qquad\hfill \text{if } \hat{d}_{t - 1}(v) \geq 2 \text{ or if $v$ is stalling} \\
											\min\{ 2p_{t - 1}(v), 1/2 \}, \hfill \text{otherwise.}
									\end{cases}
					\]

				\If{Node $v$ is not stalling}
					\State {\bf Round 2:} Node $v$ marks itself w.p. $p_{t}(v)$.
					\State If $v$ is the only marked node in $N(v)$, then $v$ joins the MIS. 
					\State If $v$ joined the MIS, $v$ is removed from the graph along with its neighbors.
				\EndIf
			\EndFor
		\EndFor
	\end{algorithmic}
\end{algorithm} 
% \ju{Due to multiple samplings, we might need to add a roughly $\log \Delta$ factor to the upper bound on the number of neighbors that are sampled.}
% \begin{remark}
% 	\label{remark:bigdeg}
% 	For the sake of our analysis, we need to handle the following unlikely event separately.
% 	Consider the event where in some iteration $t$, $d_{t}(v) = x$ and the number of marked neighbors of node $v$ is more than $\max\{ 2x, 200\log \Delta \}$.
% 	In case this event occurs for any iteration, we remove $v$ from the problem and deal with it in the post-shattering phase of the algorithm.
% \end{remark}\ju{incorporate this in the analysis.}

\paragraph{Analysis for the Sparsified Algorithm:}
We provide an analysis which shows that the above sparsified algorithm provides guarantees similar to those of the algorithm of \cite{Ghaffari2016}. The formal statement is provided below, and the proof appears in \Cref{app:thm:beeping}.
\begin{theorem}
	For each node $v$, during $T = c(\log \Delta)$ iterations for a sufficiently large constant $c$, with probability at least $1/\Delta^{C}$, either node $v$ or a neighbor of $v$ is added to the MIS. This guarantee holds independent of the randomness outside the $2$-hop neighborhood of $v$.
	Furthermore, let $B$ be the set of nodes remaining after $T$ rounds. With probability at least $1 - 1/n^{10}$, we have the following:
	\begin{enumerate}
		\item Each connected component of the graph induced by $B$ has $\bigO(\log_\Delta n \cdot \Delta^4)$ nodes.
		\item $|B| \leq \frac{n}{\Delta^{10}}$.		
		\item If $\Delta > n^{\alpha/4}$, then the set $B$ is empty.
	\end{enumerate}
	\label{thm:beeping}
\end{theorem}

% \ju{remember to mention the hand-picked nodes in Remark~\ref{remark:bigdeg}.}

\subsection{Constructing a Sparse Graph to Simulate a Phase of the Sparsified Algorithm}
\label{sec:sparsegraph}
% What we want:
% \begin{enumerate}
% 	\item nodes that get sampled in some iteration in a phase
% 	\item nodes that get marked in some iteration in a phase
% 	\item low degree
% \end{enumerate}
% Problem:
% A heavy node might get sampled. 
% Solution: heavy nodes are stalling. We can predict what they do.

% \paragraph{Sparsified Graph $\Gt$.}
We now describe how we build the sparse graph $H$ at the beginning of the phase, such that we can run $O(R)$ rounds of the sparsified algorithm on just this graph. The role of the sparse graph $H$ will be similar to the one in the warm up provided in \Cref{sec:warmup}.
   
\paragraph{Fixing All the Randomness in the Beginning:} We first draw the randomness that each node will use, at the very beginning of the execution. 
Every node $v$ draws $O(\log^3 \Delta)$ random bits such that there are $c_1 \log^2 \Delta$ bits of fresh randomness for each of the  $c \log \Delta$ iterations of the MIS algorithm, for a desirably large constant $c_1 > c$.
For iteration $t$, let $\bar{r}_t(v) = (r^1_{t-1}, \ldots, r^k_{t-1}, r^m_{t})$ denote a vector of $k + 1$ uniformly chosen random numbers, with $c_1 \log \Delta$-bit precision\footnote{In the extreme case, $p_t(v)$ of node $v$ halves in every iteration. Since $p_0(v) = 1/2$, $p_t(v)$ is a power of two in every iteration, and the number of iterations is bounded by $O(\log \Delta)$, the number of random bits needed per iteration is also $O(\log \Delta)$.}, from the interval $[0, 1]$.
Given this, once we know $p_{t}(v)$ for some iteration $t$, we can derive the outcome of the random marking for iteration $t$ by checking whether $r^{m}_t(v) < p_t(v)$. Similarly, a node is sampled in the $j$-th repetition if, in $\bar{r}_{t}(v) = (r^1_{t - 1}, \ldots, r^k_{t - 1}, r^m_t)$, we have $r^j_{t-1} < p_{t-1}(v)$. We note that once we have fixed each node's randomness as above, the behavior of the algorithm is fully deterministic.

\medskip
We denote an interval of iterations from $t$ to $t'$ by $[t, t']$ and refer to it as a phase if $t$ is the beginning and $t'$ is the end of the same phase in Algorithm~\ref{alg:localmis}. We next explain how we construct the sparsified graph $\Gt$ for phase $[t, t']$, after introducing some helper definitions.

\paragraph{Definitions.}
We use the following terminology in the construction of our sparse graph.
\begin{enumerate}
	\item  Node $u$ is \emph{relevant} if $r^{j}_{i - 1}(u) < p_{t - 1}(u) \cdot 2^{\alpha \cdot \sqrt{\log \Delta} / 10}$ for some iteration $i \in [t, t']$ and any index $1 \leq j \leq k$ or if $r^{m}_{i}(u) < p_{t}(u) \cdot 2^{\alpha \cdot \sqrt{\log \Delta} / 10}$.
	\item We say that node $u$ is \emph{light}, if $d_{t - 1}(u) < 2^{\alpha \cdot \sqrt{\log \Delta} / 5 + 1}$.
Otherwise, $u$ is \emph{heavy}.
	\item We say that $u$ is \emph{good} if the following inequality holds for all $i \in [t, t']$ and otherwise, it is bad.
		\[
			\hat{d}_{i - 1}(u) \leq 2^{\alpha \cdot (3/10) \cdot \sqrt{\log \Delta} + 2} = 4 \cdot 2^{\alpha \cdot (3/10) \cdot \sqrt{\log \Delta}} \ .
		\]
\end{enumerate}
% We say that node $u$ is \emph{relevant} if $r^{j}_{i - 1}(u) < p_{t - 1}(u) \cdot 2^{\eps \cdot \sqrt{\log \Delta} / 10}$ for some iteration $i \in [t, t']$ and any index $1 \leq j \leq k$ or if $r^{m}_{i}(u) < p_{t}(u) \cdot 2^{\eps \cdot \sqrt{\log \Delta} / 10}$.
Notice that if $u$ is not relevant, it will not get marked nor sampled in phase $[t, t']$. Hence, we do not need to include $u$ in our sparse graph. 
% We say that node $u$ is \emph{light}, if $d_t(u) < 2^{\eps \cdot \sqrt{\log \Delta} / 5 + 1}$.
% Otherwise, $u$ is \emph{heavy}.
For a light node $u$, we have that $d_{i}(u) < 2^{\alpha \cdot (\sqrt{\log \Delta} / 5 + \sqrt{\log \Delta} / 10) + 1} =  2^{\alpha \cdot (3/10) \cdot \sqrt{\log \Delta} + 1}$ for all iterations $i \in [t, t']$.
% We say that $u$ is \emph{good} if the following inequality holds for all $i \in [t, t']$ and otherwise, it is bad.
% \[
% 	\hat{d}_{i}(u) \leq 2^{\eps \cdot (3/10) \cdot \sqrt{\log \Delta} + 2} = 4 \cdot 2^{\eps \cdot (3/10) \cdot \sqrt{\log \Delta}} \ .
% \]

\paragraph{Constructing the Sparse Graph $\Gt$.}
We first determine the vertices of $\Gt$. All relevant light nodes that are good are added to $\Gt$.
For a relevant heavy node $u$, we create $d$ virtual copies, where $d$ is the number of relevant light nodes that are good and connected to $u$ in the original graph.
All these copies are added to $\Gt$. We next determine the edges of $\Gt$. If two light nodes $u$ and $w$ are connected in the original graph, we add the edge $\{ u, w \}$ to $\Gt$. Each copy of a relevant heavy node $w$ gets an edge to exactly one of the light nodes that $w$ is connected to in the original graph. Hence, every heavy node in $\Gt$ has degree one and is connected to a light node. Finally, we note that some vertices carry extra information when added to $\Gt$, which is maintained as a label on the vertex. In particular, every node $u$ in $\Gt$ is labeled with its random bits $\bar{r}_{i}(u)$ for all iterations $i$ in $[t, t']$. This label can be thought of as a bit string appended to the identifier of the node. 

\begin{observation}
	Given the $p_t(v)$ values and the random bits $\bar{r}_t(v)$ for each node $v$, the graph $\Gt$ can be constructed from the $1$-hop neighborhood of each node.
	\label{obs:construct}
\end{observation}
% \mtodo{x should change to R}
\begin{lemma}
	\label{lemma:MISsparsify}
	Let $R = \alpha \cdot \sqrt{\log \Delta} / 10$ and $C \geq 1$ a desirably large constant.
	A light node $v$ is bad in phase $[t, t']$ with probability\footnote{notice that to get a probability of $1/\Delta^{C}$, we can insert any $R \geq \log \log \Delta$.} at most $e^{-2^{2R}} \ll 1/\Delta^{C}$.
	Furthermore, the event that a node is bad is independent of the randomness of nodes outside of its $2$-hop neighborhood.
\end{lemma} 
\begin{proof}
	By definition, $d_t(v) < 2^{\alpha \cdot \sqrt{\log \Delta} / 5 + 1} = 2^{2R + 1}$.
	Since the $d_t(v)$ value increases by at most a factor of two in every iteration, we get that $d_i(v) < 2^{2R + 1} \cdot 2^{R} = 2^{3R + 1}$ for any iteration $i \in [t, t']$.
	The expected value $\expect[\hat{d}_i(v)] = \mu$ is therefore bounded from above by $2^{3R + 1}$ and thus, by a Chernoff bound, we have
	\[
		\pr\left( \hat{d}_i(v) > 2\mu \right) = \pr\left( \hat{d}_i(v) > 2^{3R + 2} \right) < e^{-2^{3R + 1} \cdot (1/3)} <  e^{-2^{3R - 1}} = e^{-2^{\alpha \cdot (3/10)\sqrt{\log \Delta} - 1}} \ .
	\]
	
	Node $v$ is bad if there is at least one iteration $i$ such that $\hat{d}_i(v) > 2^{\alpha \cdot (3/10) \sqrt{\log \Delta} + 2} = 2^{3R + 2}$.
	By a union bound over the iterations and the sampling repetitions, for a sufficiently large $\Delta$, we get that the probability of node $v$ being bad is at most 
	\[
		R \cdot O(\log \Delta) \cdot e^{-2^{3R - 1}} < e^{-2^{2R}} \ll 1/\Delta^{C} \ .
	\]
	We get the independence by observing that $2^{3R + 1}$ is an upper bound for $\expect[\hat{d}_i(v)]$ regardless of the random choices of its neighbors.
	Thus, the bad event only depends on the randomness of the neighbors of $v$ in the corresponding iteration.
\end{proof}

\begin{lemma}
	\label{lemma:degsize}
	Let $R = \alpha \cdot \sqrt{\log \Delta} / 10$ be the length of phase $[t, t']$.
	The maximum degree of $\Gt$ is $O\left( 2^{5R} \right)$.
	Furthermore, the number of nodes in the $R$-hop neighborhood of any node node $v \in \Gt$ is bounded from above by $O\left( 2^{5R} \right)^R < \Delta^{\alpha^2 /8} \ll n^{\alpha}$.
\end{lemma}
\begin{proof}
	By definition, all (copies of) heavy nodes in $\Gt$ have degree exactly $1$.
	Since we only picked good light nodes $u$, we have that $\hat{d}_{i}(u) \leq 4\cdot 2^{3x}$ for all $i \in [t, t']$.
	Let $C$ be the constant from Algorithm~\ref{alg:localmis}.
	Thus, summing over all repetitions of the sampling, the number of sampled and marked neighbors of $u$ is bounded from above by $12 \cdot C \log \Delta \cdot \hat{d}_{i}(u)$ for any iteration $i$.
	Summing up over all iterations, for a sufficiently large $\Delta$, we can bound the number of neighbors of $u$ by
	\[
		O\left( \frac{1}{\alpha} \cdot \sqrt{ \log \Delta} \right) \cdot O(\log \Delta) \cdot \hat{d}_{i}(u) = O\left( 2^{5R} \right) = O\left( 2^{\alpha \cdot \sqrt{\log \Delta}/2} \right) \ .
	\]
	For the second claim, the number of neighbors in the $R$-hop neighborhood of any node $v$ is at most 
	\[
		O\left( 2^{5R} \right)^{R} = O\left( 2^{\alpha \cdot \sqrt{\log \Delta} / 2} \right)^{\alpha \cdot \sqrt{\log \Delta}/10} = O\left(2^{\alpha^2 \cdot (1/20) \log \Delta} \right) \ll \Delta^{\alpha^2 / 8} \ . \qedhere
	\]
\end{proof}

\begin{lemma}
	Consider a phase $[t, t']$ of length $R = \alpha \cdot \sqrt{\log \Delta}/10$. 
	If node $v$ learns its $R$-hop neighborhood in $\Gt$, it can simulate its behavior in iterations in $[t, t']$.
	In particular, node $v$ learns $p_{t'}(v)$ and whether it joined the MIS or not.
	\label{lemma:localcorrectness}
\end{lemma}
\begin{proof}
	We argue by induction on the iteration index that the behavior of $v$ can be derived solely based on the nodes in $\Gt$.
	Consider first the base case, i.e., iteration $t$.
% 	Due to our labeling, node $v$ can determine whether it is stalling from its own label.
	From the labels of its neighbors, node $v$ can determine $p_{t - 1}(u)$ for each neighbor $u$.
	Combined with the random bits $\bar{r}_{t}(u) = (r^1_{t-1}, \ldots, r^k_{t-1}, r^m_{t})$, node $v$ can determine $\hat{d}_{t - 1}(v)$.
	If $\hat{d}_{t - 1}(v) > 2^{2R}$, node $v$ knows that it is stalling and hence, $p_{t'}(v) = 2^{-R} \cdot p_{t - 1}(v)$ and it will not join the MIS.
	
	Thus, we focus on non-stalling nodes for the rest of the proof.	
	By construction of $\Gt$, if $u$ is sampled or marked in iteration $t$, it belongs to $\Gt$.	
	Otherwise, the node $u$ has no impact on the behavior of $v$.
	From the $p_{t - 1}(u)$ values of its sampled neighbors, node $v$ can derive $\hat{d}_{t - 1}(v)$ and further $p_{t}(v)$.
	Once $v$ knows $p_{t}(v)$, it can derive whether it gets marked or not.
	If any neighbor of $v$ is marked, then $v$ cannot join the MIS.
	Conversely, if no neighbor is marked and $v$ is, then $v$ joins the MIS.	
% 	Furthermore, since $\Gt$ contains all the neighbors that are sampled in at least one sampling repetition of iteration $t$ and the labels of these neighbors contain the randomness used in this iteration, node $v$ can derive $\hat{d}_t(v)$ and set $p_{t + 1}(v)$ accordingly.
% 	Thus, node $v$ can determine $p_{t + 1}(v)$.
% 	The argumentation is the same for a stalling node $v$, except that a stalling node must consider the neighbors of all of its copies and it does not join the MIS even if it is marked and none of its neighbors are.
	
	Assume then that the claim holds for iteration $t \leq i < t'$.
	Due to the construction of $\Gt$ $v$ can learn the random bits by looking at the labels of its neighbors.
	By the induction hypothesis, in iteration $i$, node $v$ knows the $p_{i - 1}(u)$ values of all of its neighbors and whether they joined the MIS or not.
	In case $v$ or a neighbor joined the MIS, we are done.
	With the knowledge of the random bits $\bar{r}_i(u)$ and the $p_{i - 1}(u)$ values of its neighbors, it can simulate all the repetitions of the sampling process in iteration $i$.
	Thereby, $v$ can derive whether $\hat{d}_{i - 1}(v) \geq 2$ and set $p_{i}(v)$ accordingly.
	Thus, $v$ can determine whether it is marked or not and simulate round $2$ of iteration $i$.
% 	The remainder of the claim follows from analogous considerations as in the base case.\ju{say what to do with random bits}
\end{proof}

\subsection{Simulation in the Low Memory MPC Model}
In this section, we explain how by simulating the above algorithm, we can prove \Cref{thm:mainMPC} for the MIS problem. The extensions to the other problems follow by simple adjustments and known methods and are discussed in \Cref{subsec:implications}. 

\begin{remark}
	In case the maximum degree $\Delta > n^{\alpha}$, one needs to pay attention to how the model takes care of distributing the input.
	One explicit way is to split the high degree nodes into many copies and distribute the copies among many machines.
	For the communication between the copies, one can imagine a (virtual) balanced tree of depth at most $1 / \alpha$ rooted at one of the copies.
	Through this tree, the copies can exchange information in $O(1/\alpha)$ communication rounds.
	Another subtlety is that without care, communicating through this tree might overload the local memories of the machines.
	In our algorithms, the messages are very simple and hence, do not pose a problem.
	For the sake of simplicity, the write-up in this section assumes that all edges of each node fit within one machine's memory and in particular, $\Delta < n^{\alpha}$. The algorithm can be extended easily to higher values of $\Delta$ by doing a $O(1/\alpha)$ rounds of communication atop the virtual tree mentioned above.
\end{remark}

Our low memory MPC algorithm performs $\log \log \Delta + 1$ steps, where in step $i = 1, 2, \ldots$ we execute Algorithm~\ref{alg:localmis} on the subgraph induced by nodes with degree at least $\Delta_i=\Delta^{2^{-i}}$. This ensures that after step $i$, all vertices of degree at least $\Delta^{2^{-i}}$ are removed, and therefore, the maximum degree in the remaining graph is at most $\Delta^{2^{-i}}$. Let $n_i$ be the number of nodes in the graph in step $i$.

\begin{lemma}
	Let $v$ be a node in the graph remaining in step $i$.
	Consider phase $s$ of Algorithm~\ref{alg:localmis}, which was run in step $i$, and let $\Gs$ be the corresponding sparsified graph.
	Each node $v$ in $\Gs$ can learn its $(\alpha \cdot \sqrt{2 \log \Delta_i}/10)$-hop neighborhood in $\Gs$ in the low memory MPC model in $O(\log \log \Delta_i)$ communication rounds.
	In particular, the $(\alpha \cdot \sqrt{2 \log \Delta_i}/10)$-hop neighborhood of a node in $\Gs$ in step $i$ fits into the memory of a single machine and the neighborhoods of all nodes in $\Gs$ fit into the total memory.
	\label{lemma:learnnbh}
\end{lemma}
\begin{proof}
	Consider the following well-known and simple graph exponentiation procedure~\cite{Lenzen:2010, Ghaffari2017}.
	In every communication round, each node $u$ informs its neighbors of the nodes contained in $N(u)$.
	Then, every node can add the new nodes it learned about in its neighborhood by adding a virtual edge to each such node.
	This way, in round $j$ of the procedure, node $v$ will be informed about all nodes and edges in its $2^j$-hop neighborhood.
	Thus, every node learns its $\alpha \cdot (\sqrt{2 \log \Delta_i}/10)$-hop neighborhood after at most $O(\log \log \Delta_i)$ rounds.	

	Due to the design of the algorithm, we have that the minimum degree of a node $v$ considered in step $i$ is at least $\Delta_i$ and hence, the total memory we have is at least $O(n_i \cdot \Delta_i)$.
	Furthermore, the maximum degree is at most $\Delta^2_i$.	
	By \Cref{lemma:degsize}, this implies that the $\left( \alpha \cdot \sqrt{\log \Delta^2_i}/10 \right)$-hop neighborhood of any single node in $\Gs$ contains at most $\Delta_i^{2 \alpha^2 / 8} = \Delta_i^{\alpha^2/4}$.
	
	Hence, we need to store at most $\Delta_{i}^{\alpha^2/4}$ virtual edges per node per phase.
	Since $\Delta_{i}^{\alpha^2/4} \ll n^{\alpha}$, the neighborhood together with the virtual edges clearly fit into the memory of a single machine.
	Combined with the labels, the total memory required is then $n_i \cdot \Delta_{i}^{\alpha^2/4} \cdot O(\log^3 \Delta) = O\left( n_i \cdot \Delta_{i}^{\alpha^2/2} \right)$.
	For the next phase, we can re-use the same memory.
	We conclude that the total memory suffices to store a copy of the $(\alpha \cdot \sqrt{2 \log \Delta_i}/10)$-hop neighborhood of every node in any phase in step $i$.
\end{proof}

\begin{theorem}
	There is an algorithm that, with probability $1 - 1/n^{10}$, computes a Maximal Independent Set in the low memory MPC model that requires $O(\sqrt{\log \Delta} \log \log \Delta + \sqrt{\log \log n})$ communication rounds.
\end{theorem}
\begin{proof}
% Let $[t, t']$ be a phase of Algorithm~\ref{alg:localmis}.
% Given that we know the $p_{t - 1}(v)$ value of every node, we can construct graph $\Gt$ in a single communication round.
By \Cref{lemma:learnnbh}, we can simulate one phase of Algorithm~\ref{alg:localmis} in $O(\log \log \Delta)$ communication rounds.
Hence, we can simulate all the $O(1/\alpha) \sqrt{\log \Delta}$ phases in $O((1/\alpha) \sqrt{\log \Delta} \log \log \Delta)$ rounds.
By \Cref{thm:beeping} and \Cref{lemma:MISsparsify}, we get that with probability at most $1/\Delta^{C - 2}$, a node survives, i.e., neither it nor at least one of its neighbors is part of the MIS after executing our simulation.
Hence, we can apply \Cref{thm:beeping} and obtain that the connected components induced by the surviving nodes are of size at most $O(\Delta^4 \cdot \log n)$ and the number of the surviving nodes is at most $n^* = n/\Delta^{10}$.
By \Cref{lemma:learnnbh}, we do not break our memory restrictions.

We apply the graph exponentiation procedure once more and simulate the deterministic algorithm for MIS designed for these small components by Ghaffari~\cite{Ghaffari2016}.
This simulation requires each node to know its $2^{O(\sqrt{\log \log n})}$-hop neighborhood.
Notice that since the components are of size $O(\Delta^4 \cdot \log n)$, we require $n^* \cdot O((\Delta^4 \cdot \log n)^2) = \tilde{\Theta}(n)$ total memory to store these neighborhoods\footnote{Notice that if $\Delta^4 > n^{\alpha}$, i.e., a component does not fit into the memory of a single machine, we have by \Cref{thm:beeping} that no node survives the $\Theta(\log \Delta)$ rounds of Algorithm~\ref{alg:localmis}.}.
Hence, we obtain a runtime of $O(\sqrt{\log \log n})$ for the deterministic part.

Putting the randomized and the deterministic part together and summing up over all steps results in a runtime of
\begin{align*}
		& \sum_{i = 1}^{\log \log \Delta + 1} O \left( \frac{1}{\alpha} \sqrt{\log \Delta_i} \cdot \log \log \Delta_i \right) + O\left( \log 2^{\sqrt{\log \log n}} \right)  \\
								& \phantom{++} = O\left( \sum_{i = 1}^{\log \log \Delta + 1} \sqrt{2^{-i} \log \Delta} \cdot \log \log 2^{-i} \Delta \right) + O\left(\sqrt{\log \log n} \right) \\
								& \phantom{++}	 = O\left( \sqrt{\log \Delta} \cdot \log \log \Delta + \sqrt{\log \log n} \right) \ . \qedhere
\end{align*}
\end{proof}

\subsection{Simulation in the LCA Model}
Similarly to our MPC algorithm, our LCA algorithm for MIS simulates phases of Algorithm~\ref{alg:localmis} by creating the sparsified graph $\Gi$ for every phase $i$.
For the purposes of our LCA, we can set the length of each phase to be $\sqrt{\log \Delta} / 10$, i.e., omit the $\alpha$ factor.
% To create an LCA that simulates the second phase, we think of the LCA for the first phase as an oracle that provides a query access to the $p_t(u)$ values of nodes in the original graph $G$.
% Through this oracle, we can also figure out whether $u$ belongs to $\Gt$.
It is convenient to think about simulating a phase as creating an oracle that, for node $v$, answers the following query:
What is the state of node $v$ in the end of phase $i = [t, t']$?
In particular, did $v$ join the MIS and what is $p_{t'}(v)$.

\paragraph{Oracle $\orc_0(v)$.}
From its $1$-hop neighborhood, node $v$ can derive whether $u$ belongs to $H_0$ and whether it is stalling or not in phase $0$.
Similarly, $v$ can deduce its $1$-hop neighborhood in $H_0$ by querying every node in its neighborhood in the original graph. 
Then, iteratively, $v$ can learn its $(x + 1)$-hop neighborhood in $H_0$ by querying all the neighbors of the nodes within the $x$-hop neighborhood in $H_0$.
In particular, we do not query the neighbors of nodes that are not part of the $x$-hop neighborhood of $v$ in $H_0$.
Once the $(\sqrt{\log \Delta}/10)$-hop neighborhood is learned, we can simulate the behavior of $v$ in phase $0$ by \Cref{lemma:localcorrectness}.

\paragraph{Oracle $\orci(v)$.}
Consider phase $i > 0$ with iterations $[t, t']$.
First, we query $\orc_{i - 1}(v)$ and each neighbor $u$ of $v$, we query the state of $u$ from $\orc_{i - 1}(u)$.
In particular, we learn $p_{t - 1}(v)$, $p_{t - 1}(u)$ for all neighbors, the random\footnote{Notice that labeling the graph explicitly with the random bits $\bar{r}_t(u)$ for an LCA is not necessary due to the shared randomness, i.e., the random bits are availably by definition.} bits $\bar{r}_{t}(u)$, and whether $v$ or any neighbor $u$ joined the MIS.
From this information, we are able to derive $\hat{d}_{t-1}(v)$ and whether $v$ belongs to $\Gi$.
Then, we use the same procedure to figure out which neighbors of $v$ belong to $\Gi$.
Once we have learned the $1$-hop neighbors of $v$ in $\Gi$, we iteratively learn their neighbors in $\Gi$ until we have learned the $(\sqrt{\log \Delta}/10)$-hop neighborhood of $v$ in $\Gi$.
Once the $(\sqrt{\log \Delta}/10)$-hop neighborhood is learned, we simulate phase $i$ on the graph $\Gi$.

Let us denote the query complexity of simulating phase $i$ by $Q(i)$.
\begin{lemma}
	The oracle $\orci(v)$ for phase $i$ for node $v$ requires at most $Q(i - 1) \cdot \Delta^{9/8}$ queries.
	\label{lemma:onephasequeries}
\end{lemma}
\begin{proof}
% 	Suppose that we have query access to the oracle $\orc_{i - 1}(u)$ for all nodes $u$.
	According to the design of $\orci(v)$, we use the oracle $\orc_{i - 1}(u)$ to query each neighbor $u$ of every node in the $(\sqrt{\log \Delta}/10)$-hop neighborhood of $v$ in $H_{i - 1}$.
% 	Notice that we do not query any nodes in the $(\sqrt{\log \Delta}/10)$-hop neighborhood of $v$ in $G_{i - 1}$ and that are not adjacent to any node in this neighborhood.
	By \Cref{lemma:degsize}, we get that the number of nodes whose states are queried from $\orc_{i - 1}$ is at most $\Delta^{1/8} \cdot \Delta = \Delta^{9/8}$, i.e., $Q(i) = Q(i - 1) \cdot \Delta^{9/8}$.
\end{proof}

\begin{theorem}
	There is an LCA that, with probability $1 - 1/n^{10}$, computes an MIS with query complexity $\Delta^{O(\sqrt{\log \Delta})} \cdot \log n$.
\end{theorem}
\begin{proof}
	Let $x$ be the number of phases of Algorithm~\ref{alg:localmis} and hence $x = O(\sqrt{\log \Delta})$.
	Using \Cref{lemma:onephasequeries} can bound the total query complexity of our simulation from above by
	\begin{align*}
		Q(x - 1) & = \Delta^{9/8} \cdot Q(x - 2) = \Delta^{9/8} \cdot \Delta^{9/8} \cdot Q(x - 3) = \ldots \\
			 & = \prod_{i = 0}^{T/x - 1} \Delta^{9/8} = \Delta^{O(\sqrt{\log \Delta})} \ .
	\end{align*}
	
	After running these phases, by \Cref{thm:beeping}, the remaining components are of size at most $O(\Delta^4 \cdot \log n)$.
	Thus, we can complete our LCA for node $v$ by learning all nodes in the corresponding component resulting in a total query complexity of $Q(x - 1) \cdot \Delta^{O(1)} \cdot \log n = \Delta^{O(\sqrt{\log \Delta})} \cdot \log n$.
\end{proof}

\subsection{Implications on Other Problems}
\label{subsec:implications}
\paragraph{Maximal Matching and $2$-Approximation of Minimum Vertex Cover:}
It is well-known that a maximal matching algorithm immediately implies a $2$-approximation of minimum vertex cover, simply by outputting all endpoints of the maximal matching. 
Next, we discuss how to adjust the MIS algorithm so that it solves maximal matching. 
Our idea follows the standard approach of running an MIS algorithm on the line graph of the input graph.
In the LCA setting, we can do this without any extra effort.

However, for the MPC setting, we need some more care with the memory restriction: We make each node simulate the behavior of its edges, without creating the line graph explicitly. At first glance, it might seem that we have a problem with the local memory constraints per machine, since every edge needs to learn about up to $\Delta$ elements in its neighborhood, which amounts to $\Delta^2$ per node and exceeds our memory limitation. To overcome this issue, we make two observations: (I) We can simulate a round of sampling on the line graph without breaking the memory limit, since this only requires counting the number of sampled neighbors per edge. This can be done by exchanging one small message per edge, since an edge has endpoints in at most $2$ machines. Hence, we can derive for each edge in the beginning of a phase, whether it is part of a graph $\Gs$ or not. (II) Once we focus on the sparsified graph, by Lemma~\ref{lemma:degsize}, the maximum degree of $\Gs$ is at most $\Delta^{\alpha^2 / 8}$.
Hence, each node can simulate all of its edges in $\Gs$, including learning the information about their $\Theta(\sqrt{\log \Delta} )$-neighborhood in $\Gs$.

% Generally, this will be by running the MIS algorithm on the line graph of the maximal matching. For the LCA setting, this suffices. However, for the MPC setting, we need to pay attention to one detail: we cannot afford to store $\Delta^\alpha$ bits for every edge (as would be implied by a streaightforward application of the MIS algorithm on the line graph). In particular, we explain that it will suffice that each node simulates all of its edges, using $\Delta^{\alpha}$ bits of memory. 
% \dots
% \dots
% \mtodo{to be added \dots}

\paragraph{A $(1+\eps)$-approximation of Maximum Matching:} By a method of Mcgregor~\cite{mcgregor2005finding}, one can compute a $(1+\eps)$ approximation of maximum matching, for any constant $\eps>0$, by a constant number of calls to a maximal matching algorithm on suitably chosen subgraphs (though the dependency on $\eps$ is super exponential). These subgraphs are in fact easy to identify, and can be done in $O(1)$ rounds of the $\mathsf{LOCAL}$ distributed model. Therefore, we can use the same method to extend our maximal matching algorithm to a $(1+\eps)$-approximation of maximum matching, in both MPC and LCA, without any asymptotic increase in our complexities. We note that a similar idea was used by \cite{czumaj2017round, assadi2017coresets, ghaffari2018improved} to transform constant approximation of maximum matching to a $(1+\eps)$-approximation.  

\section{An Improved LCA for MIS}\label{sec:ImprovedLCA}
In this section, we modify Algorithm~\ref{alg:localmis} in a way that admits a much more efficient simulation in the LCA model.
Next, we explain the structure of our modified algorithm in a recursive manner.

\subsection{Recursive Splitting to Subphases}
On the highest level, we can think of $T = \Theta(\log \Delta)$ iterations in Algorithm~\ref{alg:localmis} as a (very long) phase $s_0$ of length $T$.
In our modified algorithm, any node $u$ in phase $s_0$ that has $\hat{d}_0(u) > 2^{2T}$ is stalling\footnote{The highest level phase is degenerate in the sense that $\hat{d}_0$ is potentially never larger than $2^{2T} = \poly \Delta$. However, for the sake of presentation, it is convenient to start from the largest possible phase.} and thus, will not be marked and halves its $p_t(u)$ value in every iteration $t$ of $s_0$.
For the non-stalling nodes, we split the phase of $T$ iterations into two subphases of $T/2$ iterations.
In the subphases of length $T/2$, we adjust the threshold for stalling to $\hat{d} > 2^{2T/2} = 2^{T}$.
After recursively splitting the (sub-)phases $i$ times, we reach subphases of length $R = T / 2^i$.
In subphases of length $R$, node $u$ is stalling if $\hat{d}(u) > 2^{ 2R }$ in the first iteration of the phase.
The recursive splitting to subphases is continued until we hit a subphase length of $2\log \log \Delta < R \leq 4\log \log \Delta$.
%Notice that once the length of a subphase $s$ is less than $4 \log \log \Delta$, we have a $\poly \log \Delta$ threshold for stalling.

\subsection{Bound on the Number of Iterations}
The estimation of $d_{t - 1}(v)$ through $\hat{d}_{t - 1}(v)$ and updating of $p_t(v)$ in iteration $t$ is done exactly as in Algorithm~\ref{alg:localmis}.
We refer to the modified version of Algorithm~\ref{alg:localmis} to as the \emph{recursive MIS algorithm}.
To bound the number of iterations that the recursive MIS algorithm needs to perform, we can use an analysis that is almost exactly the same as for Algorithm~\ref{alg:localmis}. The formal statement is presented below and the proof appears in \Cref{app:thm:sparseiterations}.
The number of iterations $T$ executed by the recursive MIS algorithm is equal to the length of the highest level phase $s_0$.
% The only difference to Algorithm~\ref{alg:localmis} is in the stalling behavior, i.e., we change the thresholds when a node stalls.
% By examining the proof of Lemma~\ref{lemma:runtime}, we can use the same arguments to obtain the guarantees of Theorem~\ref{thm:beeping} (for $\eps = 1$) for the modified algorithm.\ju{referring to the correct theorem? Do we care about the $\eps = 1$?}
% A minor difference is in the choice of numbers in the amortized argument in Lemma~\ref{lemma:runtime}, where it is easy to verify that it also works for the modified version.
\begin{theorem}
	Consider the recursive MIS algorithm described above.
	For each node $v$, during $T = c\log \Delta$ iterations for a sufficiently large constant $c$, with probability at least $1 - 1/\Delta^{98}$, either node $v$ or a neighbor of $v$ is added to the MIS. 
	This guarantee holds independent of the randomness outside the $2$-hop neighborhood of $v$.
	Furthermore, let $B$ be the set of nodes remaining after $T$ rounds.
	Each connected component of the graph induced by $B$ has $O(\log_\Delta n \cdot \Delta^4)$ nodes.
	\label{thm:sparseiterations}
\end{theorem}

\subsection{Sparsification}
Intuitively, splitting the execution of the algorithm into very short phases leads to simulating iterations on very sparse subgraphs.
In these sparse graphs, in terms of query complexity, it is cheap to simulate the iterations of the short phases.

\paragraph{Definitions.}
Here, we use terminology very similar to \Cref{sec:sparsegraph}.
For iteration $t$, let $\bar{r}_t(v) = (r^1_{t-1}, \ldots, r^k_{t-1}, r^m_{t})$ denote a vector of $k + 1 = \Theta(\log \Delta)$ uniformly chosen random numbers, with $c \log \Delta$-bit precision, from the interval $[0, 1]$.
Now, we can derive the outcome of the random marking for iteration $t$ by checking whether $r^{m}_t(v) < p_t(v)$ and similarly, a node is sampled in the $j$-th repetition if, in $\bar{r}_{t}(v) = (r^1_{t - 1}, \ldots, r^k_{t - 1}, r^m_t)$, we have $r^j_{t-1} < p_{t-1}(v)$.
We slightly adjust the definitions of node types to incorporate the varying lengths of (sub-)phases.

For a phase $[t, t']$ of length $R$, we have the following definitions:
\begin{enumerate}
	\item We say that node $u$ is \emph{relevant} if $r^{j}_{i - 1}(u) < p_{t - 1}(u) \cdot 2^{R}$ for some iteration $i \in [t, t']$ and any index $1 \leq j \leq k$ or if $r^{m}_i(u) < p_{t}(u) \cdot 2^{R}$.
	\item We say that node $u$ is \emph{light}, if $d_{t - 1}(u) < 2^{2R + 1}$.
Otherwise, $u$ is \emph{heavy}.
	\item We say that $u$ is \emph{good} if $\hat{d}_{i - 1}(u) \leq 2^{3R + 2}$ for all $i \in [t, t']$ and otherwise, it is bad.
\end{enumerate}
Notice that if $u$ is not relevant, it will not get marked nor sampled in phase $[t, t']$. Hence, we do not need to include $u$ in our sparse graph. 
For a light node $u$, we have that $d_{i - 1}(u) < 2^{2R + R + 1} =  2^{3R + 1}$ for all iterations $i \in [t, t']$.

\paragraph{Constructing the Sparse Graph $\Gt$.}
We first determine the vertices of $\Gt$. All relevant light nodes that are good are added to $\Gt$.
For a relevant heavy node $u$, we create $d$ virtual copies, where $d$ is the number of relevant light nodes that are good and connected to $u$ in the original graph.
All these copies are added to $\Gt$. We next determine the edges of $\Gt$. 
If two light nodes $u$ and $w$ are connected in the original graph, we add the edge $\{ u, w \}$ to $\Gt$. Each copy of a relevant heavy node $w$ gets an edge to exactly one of the light nodes that $w$ is connected to in the original graph. 
Hence, every heavy node in $\Gt$ has degree one and is connected to a light node. 
Finally, we note that some vertices carry extra information when added to $\Gt$, which is maintained as a label on the vertex. 
In particular, every node $u$ in $\Gt$ is labeled $p_{t - 1}(u)$ and the random bits $\bar{r}(v)$.
Notice that in case of an LCA, the shared randomness is available to all nodes even without an explicit labeling of $\Gt$.

The next two lemmas follow from setting $R \geq 2\log\log \Delta$ in the proof of \Cref{lemma:MISsparsify} and fixing the phase length to $R$ in the proof of \Cref{lemma:localcorrectness}.
\begin{lemma}
	\label{lemma:fineMISsparsify}
	A light node $v$ is bad in phase $[t, t']$ of length $R$ with probability at most $e^{-2^{2R}} \ll 1/\Delta^{100}$.
	Furthermore, the event that a node is bad is independent of the randomness of nodes outside of its $2$-hop neighborhood.
\end{lemma} 

\begin{lemma}
	Consider a phase $[t, t']$ of length $R$. 
	If node $v$ learns its $R$-hop neighborhood in $\Gt$, it can simulate its behavior in iterations in $[t, t']$.
	In particular, node $v$ learns $p_{t'}(v)$ and whether it joined the MIS or not.
	\label{lemma:finecorrectness}
\end{lemma}

\begin{lemma}
	\label{lemma:lowestdegsize}
	Consider a phase of length $R \geq 2 \log \log \Delta$.
	The maximum degree of $\Gt$ is at most $2^{5R}$.
	Furthermore, the number of nodes in the $R$-hop neighborhood of any node node $v \in \Gt$ is bounded from above by $2^{5R^2}$.
\end{lemma}
\begin{proof}
	By definition, all (copies of) heavy nodes in $\Gt$ have degree exactly $1$.
	Since we only picked good light nodes $u$, we have that $\hat{d}_{i - 1}(u) \leq 4\cdot 2^{3R}$ for all $i \in [t, t']$.
	Thus, summing over all repetitions of the sampling and all iterations in $[t, t']$, for sufficiently large $\Delta$, the number of sampled and marked neighbors of $u$ is bounded from above by $C \log \Delta \cdot 4 \cdot 2^{3R} \ll 2^{3R} \cdot \log^3 \Delta \leq 2^{5R}$ for any iteration $i$.
	The number of neighbors in the $R$-hop neighborhood of any node $v$ is at most $\left( 2^{5R} \right)^{R} = 2^{5R^2}$.
\end{proof}

% We want to make use of the recursive structure of Algorithm~\ref{alg:finemis}.
% Intuitively, if a node $v$ is not stalling in phase $[t, 4t]$, there are many neighbors of $v$ that are not getting marked in this phase.
% Thus, our simulation can focus on a sparser graph.
% However, since the subphase $[t, 2t]$ is much shorter than phase $[t, 4t]$, we can sparsify the graph even further for the duration of the first $2t$ iterations of phase $[t, 4t]$.
% Furthermore, if neighbor of $v$ that does not get marked or sampled in phase $[t, 4t]$, for sure will not get marked in the subphase.
% Hence, an oracle for subphase $[t, 2t]$ can be constructed only based on $H_{t, 4t}$.

A key observation in our LCA algorithm is that the behavior of a node in the two subphases of a phase\footnote{Assume w.l.o.g. for simplicity that phase lengths are multiples of $2$.} $s = [t, t']$ only depends on the graph $\Gs$.
Hence, it is convenient to think that an oracle simulating phase $s$ of length $> 4 \log \log \Delta$ answers queries to the adjacency lists of $H_{t, t'/2}$ and $H_{t'/2 + 1, t'}$.
In a sense, the oracle for phase $s$ creates the graphs $H_{t, t'/2}$ and $H_{t'/2 + 1, t'}$.

\paragraph{Oracle $\orc_{t'}(v, \Gt)$ for a Phase of Length $R$.}
If $R \leq 4 \log \log \Delta$, then $\orc_{t'}(v, \Gt)$ learns the $R$-hop neighborhood of $v$ in $\Gt$ and simulates $R$ iterations of phase $[t, t']$.
Hence, by \Cref{lemma:finecorrectness}, we obtain $p_{t'}(v)$ and the knowledge of whether $v$ joined the MIS or not.
If $R > 4\log \log \Delta$, let $s_1$ and $s_2$ be the subphases of length $R/2$ of phase $[t, t']$.
Then, the oracle $\orc_{t'}(v, \Gt)$ answers adjacency queries to $H_{s_1}$ and $H_{s_2}$.
For a query to $H_{s_1}$, we examine the neighbors of $v$ and since $\Gt$ is labeled with the $p_{t - 1}(u)$ values, we can derive the estimate $\hat{d}_{t-1}(u)$ for all neighbors $u$ of $v$ from the $2$-hop neighborhood of $v$.
If $u$ is good and light in phase $s_1$, then $u \in H_{s_1}$.
For queries to $H_{s_2}$, we query every $2$-hop neighbor $u$ of $v$ in $\Gt$ with the oracle $\orc_{t'/2}(u, H_{s_1})$.
Once we obtained $p_{t'/2}(u)$ for each $2$-hop neighbor $u$, we can derive $\hat{d}_{t'/2}(w)$ for each $1$-hop neighbor $w$ and hence, can decide whether $w \in H_{s_2}$ or not.

\paragraph{Shorthand Notation.}
We denote the number of queries needed to simulate a phase $s = [t, t']$ of length $R$ by $Q(R)$.
In other words, $Q(R)$ denotes the number of queries required by the oracle $\orc_{t'}(v, \Gt)$.
\begin{lemma}
% 	Let $Q(R)$ be the query complexity of simulating the phase $[t, t']$ of length $R$ of Algorithm~\ref{alg:finemis}.
	If $R \leq 4 \log \log \Delta$, then $Q(R) = 2^{O(\log^2 \log \Delta)}$.
	Otherwise, the query complexity $Q(R) = O\left( Q(R/2)^2 \cdot 2^{10R} \right)$.
	\label{lemma:recurse}
\end{lemma}
\begin{proof}
	Consider the phase $[t, 2t']$ of length $R$.
	Let us first examine the case where $R \leq 4 \log \log \Delta$.
	By \Cref{lemma:lowestdegsize}, any node $v$ has at most $2^{O( \log^2 \log \Delta )}$ nodes in its $(4 \log \log \Delta)$-hop neighborhood in $\Gt$.
	Hence, $Q(R) = 2^{O(\log^2 \log \Delta)}$.
	
	Consider then the case where  $R > 4 \log \log \Delta$.
	When simulating phase $[t, 2t']$, we can read the $p_{t - 1}(v)$ values from the node labels in $H_{t, 2t'}$.
	To obtain the adjacency list of $v$ in $H_{t, t'}$, we need to query all neighbors of $v$ in $H_{t, 2t'}$ once.
	Thus, by \Cref{lemma:lowestdegsize}, we need at most $Q(R/2) \cdot 2^{5R}$ queries to simulate phase $[t, t']$ and to obtain $p_{t'}(v)$.
	
	For every adjacency query to a node $u$ in phase $[t' + 1, 2t']$, we need to query $\orc_{t'}(w, H_{t, t'})$ for each neighbor $w$ of $u$.
	For each of the $Q(R/2)$ queries of phase $[t' + 1, 2t']$, we need an adjacency query to $\orc_{t'}(w, H_{t, t'})$.
	Hence, the query complexity of simulating phase $[t, 2t']$ for node $v$ is bounded by
	\[
		Q(R) = Q(R/2) \cdot \left( 2^{5R} \cdot \left( 2^{5R} \cdot Q(R/2) \right) \right) = Q(R/2)^2 \cdot 2^{10R} \ . \qedhere
	\]
\end{proof}

% \ju{phrasing of the theorem?}
% \begin{theorem}
% 	There exists an LCA that computes an MIS with query complexity $\Delta^{O(\log \log \Delta)} \cdot \log n$.
% \end{theorem}
\LCAmain*
\begin{proof}
	Let $T = O(\log \Delta)$.
	By \Cref{lemma:recurse}, we can write the query complexity of our simulation in the LCA model as
	\begin{align*}
		Q(T) & = O\left( 2^{10T} \cdot Q(T/2) \cdot Q(T/2) \right) = O\left( 2^{10T} \cdot Q(T/2)^2 \right) \\
				& O\left( \leq 2^{10 T + 10 \cdot 2 \cdot T/2} \cdot Q(T / 4)^{4} \right) = O\left( 2^{10 \cdot 2T} \cdot Q(T / 4)^{4} \right) = O\left( 2^{10 \cdot 3T} \cdot Q(T / 8)^{8} \right) = \ldots \\
				& \leq 2^{10 \left( \sum_{i = 1}^{\log T} T \right) } \cdot Q(4 \log \log \Delta)^{O(T / \log T)} \leq 2^{O(\log \Delta \log \log \Delta) } \cdot \left( 2^{O( \log^2 \log \Delta )} \right)^{O(\log \Delta / \log \log \Delta)} \\
				& = \Delta^{O(\log \log \Delta)} \cdot \Delta^{O(\log \log \Delta)} = \Delta^{O(\log \log \Delta)} \ .
	\end{align*}
	
	By \Cref{lemma:fineMISsparsify} and by \Cref{thm:sparseiterations} have that, with probability at most $1/\Delta^{98}$, a node survives, i.e., is not part of the MIS nor has a neighbor in the MIS after executing $O(\log \Delta)$ iterations.
	From \Cref{thm:sparseiterations} we have that the surviving nodes form connected components of size at most $\Delta^{O(1)} \cdot \log n$.
	Hence, we can complete our LCA for node $v$ by querying all the nodes in the corresponding component, resulting in a total query complexity of $\Delta^{O(\log \log \Delta)} \cdot \Delta^{O(1)} \cdot \log n = \Delta^{O(\log \log \Delta)} \cdot \log n$.
\end{proof}

% \section{Sketch}
% \input{sketch}

% \section{Sparsified Beeping MIS algorithm}
% \input{sparseMIS}

% \newpage
%\input{lcaMIS}

%\input{temp}

\paragraph{Acknowledgment:} Over the past year, we have discussed the notion of \emph{locality volume} with several researchers, including Sebastian Brandt, Juho Hirvonen, Fabian Kuhn, Yannic Maus, and Jukka Suomela, and we thank all of them. These discussions were typically in terms of characterizing general trade-offs between locality radius and locality volume for arbitrary locally checkable problems (though, not in the context of the problems discussed here), which we believe is an interesting topic and it deserves to be studied on its own. We hope that connections presented here to MPC and LCA settings add to the motivation.
\bibliographystyle{alpha}
\bibliography{references}

% \medskip
% \medskip
% \section{Model}
% \input{model}

 \appendix
\section{Missing Proofs}
\subsection{Proof of \Cref{thm:beeping}}
\label{app:thm:beeping}
We denote the set of nodes that are stalling in iteration $t$ by $\sh$.

Following that, we set
\[
	d_t'(v) = \sum_{u \in N(v), \ d_{t}(u) \leq 20, \ u \not \in \sh} p_t(u) \ .
\]

We define \emph{golden rounds} for node $v$.
\begin{enumerate}
	\item \textbf{Golden round type $1$:} $p_{t}(v) = 1/2, v \not \in \sh$, and $d_{t}(v) \leq 20$
	\item \textbf{Golden round type $2$:} $d_{t}(v) \geq 0.2$ and $d_{t}'(v) \geq 0.1 \cdot d_{t}(v)$.
\end{enumerate}
\begin{observation}
	In each golden round, node $v$ gets removed with a constant probability.
	\label{obs:golden}
\end{observation}

\begin{lemma}
	\label{lemma:notoff}
	Let $C$ be the constant in Algorithm~\ref{alg:localmis}, where $k = 12C \log \Delta$ is the number of repetitions of sampling.
	With probability at least $1 - 1/\Delta^{C}$, in iteration $t$
	\begin{enumerate}
		\itemsep0pt
		\item if $d_t(v) > 20$, we have $\hat{d}_t \geq 2$
		\item $\hat{d}_t(v) \leq 4d_{t}(v)$
		\item if $d_t(v) < 0.4$, we have $\hat{d}_t < 2$
	\end{enumerate}
\end{lemma}
\begin{proof}
	Suppose that $d_t(v) > 20$.
	Then, for each $j\in\{1, 2,\dots, k\}$, we have that $\expect[\hat{d}^j] = d_t(v) > 20$.
	By a Chernoff bound, we get that 
	\[
		\pr\left[ \hat{d}^j < 2 \right] = \pr\left[ \hat{d}^j < \left( 1 - \frac{9}{10} \right) \cdot \expect\left[ \hat{d}^j \right] \right] \leq e^{-0.81 \cdot 10} = e^{-8.1} \ll \frac{1}{8} \ .
	\]
	In other words, the expected number of entries in $\hat{d}$ that are larger than $2$ is less than $k/8$.
	Notice that if the median of $\hat{d}$ is less than $2$, then more than half of its entries are smaller than $2$.
	Thus, by applying a Chernoff bound, the probability that $\hat{d}_{t} \leq 2$ is at most $e^{(9 \cdot k/8)/3} < e^{k/2} < 1/\Delta^{C}$. 
	This proves the first claim.
	
	For the second claim, Markov's inequality gives that $\pr\left[ \hat{d}^{\,j} > 4 d_t(v) \right] < 1/4$.
	In other words, the expected number of entries in $\hat{d}$ that are greater than $d_t(v)$ is at most $k/4$.
	By applying a Chernoff bound, we get that the probability that $\hat{d}_{t} \geq 4d_{t}(v)$ is at most $e^{(1/3) \cdot (k/4)} < 1/\Delta^{C}$.
	%Thus, the second claim follows.
	
	The third claim follows from the second claim because if $d_t(v) < 0.4$, we have $\hat{d}_t(v) \leq 4d_t(v) < 1.6 < 2$ with probability at least $1 - 1/\Delta^{C}$.
	
% 	Let $\mu = 800 \log \Delta \cdot d_{t}(v)$ and let $\bar{m}$ the number of sampled neighbors of $v$ in round $1$ of iteration $t + 1$.
% 	According to the design of our algorithm, expected number of marked neighbors $\expect[ \bar{m} ] = \mu$.
% 	By applying a Chernoff bound, we get that 
% 	\begin{align*}
% 		\pr \left( \bar{m} > 2 \cdot  \mu \right) & < 2^{-\left( \frac{\mu}{3} \right)} < \frac{1}{\Delta^{100}} \\
% 		\pr \left( \bar{m} < \frac{1}{2} \cdot \mu \right) 	& < 2^{-\left( \frac{1}{4} \cdot \frac{\mu}{2} \right)} < \frac{1}{\Delta^{100}} \qedhere
% 	\end{align*}
\end{proof}

\begin{lemma}
	For each node $v$, during $T = c \cdot \log \Delta$ rounds, where $c$ is a sufficiently large constant, with probability at least $1 - 1/ \Delta^{C - 2}$, there are at least $0.05 \cdot T$  golden iterations.
	\label{lemma:runtime}
\end{lemma}
\begin{proof}
Let us denote the count of golden iterations of type $1$ and $2$ by $g_1$ and $g_2$, respectively.
Let $h$ denote the number of iterations in which $d_t(v) > 0.4$ or $v \in \sh$. Next, we argue that either $g_1$ or $g_2$ must be at least $0.05T$. 

\paragraph{Small $g_2$ implies small $h$.}
Assume $g_2 < 0.05T$.
% The rough idea is that only the golden rounds of type $2$ and wrong moves increase the $d_t(v)$ value.
% In other words, bounding these counts forces the desire to be small implying a small $h$.
We first analyze iterations in which $v \in \sh$ and $d_t(v) \leq 0.4$.
Consider the iteration $i \leq t$ in which $v$ started stalling such that $t < i + \alpha \cdot \sqrt{\log \Delta} / 10 = t'$.
By definition of stalling we have $\hat{d}_{i - 1}(v) \geq 2^{\alpha \cdot \sqrt{\log \Delta}/5}$.
By Lemma~\ref{lemma:notoff}, $d_{i - 1}(v) \geq 2^{\alpha \cdot \sqrt{\log \Delta} / 5 - 2}$ with probability at least $1 - 1/\Delta^{C}$.
Thus, in iteration $t'$, we have 
\[
	d_{t'}(v) \leq 0.4 \cdot 2^{\alpha \cdot \sqrt{\log \Delta} / 10 + 1} \leq 0.4 \cdot \left( 2^{\alpha \cdot \sqrt{\log \Delta} / 5 - 2} \cdot 2^{\alpha \cdot (- \sqrt{\log \Delta} / 10) + 3} \right)  < d_{i - 1}(v) \cdot 2^{\alpha \cdot (-\sqrt{\log \Delta} / 10) + 3}   \ .
\]
% with probability at least $1 - /\Delta^{99}$.
Hence, amortizing over the $t' - i = \alpha \cdot \sqrt{\log \Delta} / 10$ iterations, we have $d_{j + 1}(v) \leq 0.65 \cdot d_{j}(v)$ for all $i \leq j < t'$.
For the sake of the analysis, we may thus assume that for all iterations $t$ in which $v \in \sh$ and $d_{t}(v) \leq 0.4$, we have $d_{t + 1}(v) \leq d_{t}(v) \cdot 0.65$.
% By a union bound over all iterations, this is true with probability at least $1 - 1/\Delta^{99}$.

Consider then an iteration $t$ in which $d'_t(v) < 0.1 d_t(v)$.
In this case, $0.9 d_t(v)$ is contributed by neighbors $u$ of $v$ that are either stalling or have $d_t(u) > 20$. We argue that in this case, with probability $1 - 1/\Delta^{98}$, we will have that $d_{t + 1}(v) \leq (0.45 + 0.2) \cdot d_t(v) = 0.65 d_t(v)$.
The reason is as follows:
We can use Lemma~\ref{lemma:notoff} and the union bound over all neighbors of $v$ and all iterations to obtain that, with probability $1 - 1/\Delta^{C - 2}$,  all such ``heavy'' neighbors of $v$ in all $O(\log \Delta)$ iterations correctly detect that they are heavy (or stalling) and thus, their $p_t(u)$ value drops by a factor of $1/2$ in each such iteration.
For the other neighbors that contribute the remaining $0.1 d_t(v)$, the worst case is that they all double their $p_t$ value.
Hence, we get that $d_{t + 1}(v) \leq (0.45 + 0.2) \cdot d_t(v) = 0.65 d_t(v)$.

The above implies that the $d_t(v)$ drops by a factor $0.65$ in every iteration in $h$ that is not a $g_2$ iteration. 
Now in every $g_2$ iteration, $d_t(v)$ can increase by at most a $2$ factor.
This implies that $h \leq 3g_2  + 4\log \Delta$. 
Suppose towards contradiction that $h > 3g_2  + 4\log \Delta$. 
Then we would have 
\[
	d_t(v) < \left(0.65\right)^{h-g_2} \cdot 2^{g_2} \cdot \frac{\Delta}{2} < \left( \frac{1}{2} \right)^{2 \log \Delta} \cdot \frac{\Delta}{2} < 0.4 \ .
\]
That is, we cannot have $d_{t}(v)$ remain above $0.4$ for more than $3g_2  + 4\log \Delta$ iterations. Since we have assumed $g_2 < 0.05 T$, we conclude that $h< 0.2T$.

% Consider then the rounds in which $d_t(v) > 0.2$ that are part of the $h$ count.
% We split the examination further to two types of rounds, namely
% \begin{enumerate}[leftmargin=2\parindent]
% 	\item[\textbf{(A1)}] $d_t(v) > 0.01$ and $d'_t(v) \geq 0.01d_t(v)$. This is the same as a golden round type $2$.
% 	\vspace{+3pt}
% 	\item[\textbf{(A2)}] $d_t(v) > 0.01$ and $d'_t(v) < 0.01d_t(v)$.
% \end{enumerate}
% Since we assumed that there are at most $0.06T$ wrong moves and that $g_2 < 0.03T$, it follows that the sum of number of $(A1)$-rounds and $(A2)$-rounds with a wrong move is bounded from above by $(0.06 + 0.03)T = 0.09T$.
% Furthermore, in every $(A2)$ round \emph{without} a wrong move, $d_t(v)$ decreases by a factor of $0.6$.

\paragraph{Small $h$ implies large $g_1$.}
Suppose that $h < 0.2T$. Then, with probability $1 - 1/\Delta^{C - 1}$, in at most $0.2T$ iterations, we have $p_{t}(v)$ decrease by a $1/2$ factor. 
Besides these, in every other iteration, we have $p_t(v) = \min \{2p_t(v), 1/2 \}$. 
% Again, by Lemma~\ref{lemma:notoff} and a union bound over all iterations, the probability that these changes in the $p_t(v)$ value happen is at least $1 - 1/\Delta^{99}$.
Since we always have $p_{t}(v)\leq 1/2$, among these, at most $0.2T$ iterations can be iterations where $p_{t}(v)$ increases by a $2$ factor. 
Hence, there are at least $(1 - 2 \cdot 0.2)T = 0.6T$ iterations in which $p_t(v) = 1/2$. By assumption, the number of rounds in which $p_t(v) = 1/2$ and $d_t(v) > 0.4$ or $v \in \sh$ can be at most $h$. Therefore, we have at least $(0.6 - 0.2)T = 0.4T$ iterations in which $p_t(v) = 1/2$, $d_t(v) \leq 0.4$, and $v \not \in \sh$.
By definition, any such iteration is a golden iteration of type $1$.
Hence, we conclude that $g_1 \geq 0.4T > 0.05T$.
\end{proof}

\begin{proof}[Proof of Theorem~\ref{thm:beeping}]
	Suppose that $\Delta^4 > n^{\alpha}$ and set $C = 60/\alpha$.
	We have by \Cref{lemma:runtime} that the probability of a node remaining after executing the Algorithm~\ref{alg:localmis} for $\Theta((1 / \alpha) \cdot \log \Delta)$ iterations is at most $1/\Delta^{60/\alpha - 2} < 1/\Delta^{50/\alpha} < 1/n^{12}$.
	Now, by using a union bound over all nodes, we get that the set $B$ is empty with probability at least $1/n^{10}$.
		
	Suppose then that $\Delta^4 \leq n^{\alpha}$.
	We have that the expected number of surviving nodes is at most $n / \Delta^{60/\alpha} < n / \Delta^{60}$.
	Since the probability of remaining in the graph is independent of the randomness outside its $2$-hop neighborhood, the event of node remaining can depend on at most $\Delta^5$ other nodes.
	By using a standard variant of a Chernoff bound for bounded dependencies~\cite{Pemmaraju2001}, we get that with probability at least $1 - 1/n^{10}$, the number of nodes in $B$ is at most $n/\Delta^{10}$.
	The bound for the size of the components follows directly from previous work~\cite[Lemma 4.2]{Ghaffari2016}.	
\end{proof}

\subsection{Proof of \Cref{thm:sparseiterations}}
\label{app:thm:sparseiterations}
\begin{proof}[Proof Sketch of \Cref{thm:sparseiterations}]
	There is a small difference in the proof of this theorem as compared to the proof of Theorem~\ref{thm:beeping}.
	Due to the different stalling behavior, we need to slightly adjust the details of the proof of Lemma~\ref{lemma:runtime}.
	Consider the count $h$ that counts the number of iterations $t$ in which node $v$ is either stalling or $d_t(v) > 0.4$.
	For the case where $v$ is stalling in phase $[t, t']$ of length $R$ and $d_i(v) \leq 0.4$ for some iteration $i \in [t, t']$, we need to adjust our argument as follows.
	In the beginning of the phase, we have by Lemma~\ref{lemma:notoff} that $d_{t - 1}(v) \geq \hat{d}_{t - 1}(v) / 4 > 2^{2R - 2}$.
	Since $d_{t}(v)$ at most doubles in every iteration, we have $d_{t'}(v) \leq 0.4 \cdot 2^{R} < 2^{2R - 2} \cdot 2^{-R - 1} < d_{t - 1}(v) \cdot 2^{-R - 1}$.
	Hence, amortizing over the $R$ iterations of the phase, we have that $d_{i}(v) < d_{i - 1}(v) \cdot 0.65$ for each iteration $i$ of the phase.
	The rest of the proof is analogous to the one of Theorem~\ref{thm:beeping}.
\end{proof}

\end{document}